\documentclass[a4paper,11pt]{article}
\makeatletter
\def\ps@headings{%
\def\@oddhead{\mbox{}\scriptsize\rightmark \hfil \thepage}%
\def\@evenhead{\scriptsize\thepage \hfil \leftmark\mbox{}}%
\def\@oddfoot{}%
\def\@evenfoot{}}
\makeatother 

\usepackage[lined,boxed,commentsnumbered, ruled]{algorithm2e}
\usepackage{verbatim}
\usepackage{amsmath}
\usepackage{latexsym}
\usepackage{bm}
\usepackage{amssymb}
\usepackage{fancyhdr}
\usepackage[dvips]{graphicx}
\usepackage{subfigure}
\usepackage{amsbsy}
\usepackage{amsthm}
\usepackage{nccmath}
\usepackage{amssymb}
\usepackage{listings}
\usepackage{xcolor}
\usepackage{longtable}
\usepackage{url}
\usepackage{wasysym}
\usepackage{multirow}
\usepackage{cite}
\usepackage{setspace}
\usepackage{authblk}

 \onecolumn

\title{Heavy Traffic Limits for GI/H/n Queues: Theory and Application}

\begin{document}
\bibliographystyle{IEEEtr}

\author[*]{Yousi Zheng}
\author[*+]{Ness B. Shroff}
\author[+]{Prasun Sinha}
\affil[*]{Department of Electrical and Computer Engineering,}
\affil[+]{Department of Computer Science and Engineering, \authorcr
The Ohio State University, Columbus, Ohio 43210, USA}




\maketitle


%

%

\begin{abstract}
We consider a GI/H/n queueing system. In this system, there are
multiple servers in the queue. The inter-arrival time is general and
independent, and the service time follows hyper-exponential
distribution. Instead of stochastic differential equations, we
propose two heavy traffic limits for this system, which can be
easily applied in practical systems. In applications, we show how to
use these heavy traffic limits to design a power efficient cloud
computing environment based on different QoS requirements.
\end{abstract}

\newtheorem{theorem}{Theorem}
\newtheorem{corollary}[theorem]{Corollary}
\newtheorem{lemma}[theorem]{Lemma}
\newtheorem{proposition}[theorem]{Proposition}
\newtheorem{remark}[theorem]{Remark}
\newtheorem{claim}[theorem]{Claim}

\section{Introduction}
\label{intro}

Many large queueing systems, like call centers and data centers,
contain thousands of servers. For call centers, it is common to have
500 servers in one call center \cite{call_center}. For data centers,
Google has more than 45 data centers as of 2009, and each of them
contains more than 1000 machines \cite{google_platform}. When the
number of servers goes to infinity, many queueing systems should be
stable as long as the traffic intensity $\rho_n<1$ (i.e., the
arrival rate is smaller than the service capacity). The traffic
intensity for a queueing system with $n$ servers can be thought of
as the rate of job arrivals divided by the rate at which jobs are
serviced. At the same time, the queueing systems should work
efficiently, which means that $\rho_n$ should approach 1, i.e.,
$\lim\limits_{n\rightarrow\infty}\rho_n=1$. This regime of operation
is called the \emph{heavy traffic regime}. Our paper focuses on
establishing heavy traffic limits, and using these limits to design
a power efficient cloud based on different QoS requirements.

Some classical results on heavy traffic limits are given by Iglehart
in \cite{iglehart1965}, Halfin and Whitt in \cite{whitt1981}, and
summarized by Whitt in Chapter 5 of his recent book
\cite{whitt_book}. This heavy traffic limit ($(1-\rho_n)\sqrt{n}$
goes to a constant as $n$ goes to infinity) is now called the
Halfin-Whitt regime. Recently, the behavior of the normalized queue
length in this regime has been studied by A. A. Puhalskii and M. I.
Reiman \cite{ph}, J. Reed \cite{reed}, D. Gamarnik and P. Momeilovic
\cite{steady}, and Ward Whitt \cite{whitt2005,diffusion}. Based on
these studies, some design and control policies are proposed in
\cite{impatient,dim,dnc,level}.

Our work differs from prior work in three key aspects. First,
literature on heavy traffic limits that is based on analysis of call
center systems does not capture various unique features of large
queueing systems today, such as the cloud computing environment.
Many of those works assume a Poisson arrival process and exponential
service time \cite{impatient,dim,dnc,level}. Perhaps appropriate for
smaller systems, these models need to be generalized for today's
larger systems such as increasingly complex call-centers and cloud
computing environments. The arrival process in such complex and
large systems may be independent, but more general. More
importantly, the service times of jobs are quite varied and unlikely
to be accurately modeled by an exponential service time
distribution. In \cite{whitt2005}, Whitt also considers the
hyper-exponential distributed service time, but only with two stages
and where one of them always has zero mean. Second, although some
QoS metrics (especially Quality-Efficiency-Driven (QED)) have been
extensively studied in some call center scenarios
\cite{impatient,QED,dim}, the QoS requests can be more complex,
because of the wide variety of application needs, especially in the
cloud computing environment \cite{Berkeley,Market,Define}. And
third, while there are studies that give heavy traffic solutions for
more general scenarios \cite{ph,reed}, these solutions can only be
described by complex stochastic differential equations, which are
quite cumbersome to use and provide little insight.

\emph{In this paper, we build a system model for general and
independent inter-arrival process and hyper-exponentially
distributed service times.} As mentioned earlier, the general
arrival process can be used to characterize a variety of arrival
distributions for the queueing system. The main motivation for
studying the hyper-exponential distribution is that it can capture
the high degree of variability in the service time. For example, the
hyper-exponential distribution can characterize any
\emph{coefficient of variation} (standard deviation divided by the
mean) greater than $1$. Since the service time of jobs is expected
to be highly variable from job to job, the hyper-exponential
distribution is well suited to model the service times for today's
queueing systems.

To satisfy the QoS and save operation cost at the same time, we
characterize the performance of the queueing system for four
different types of QoS requirements: Zero-Waiting-Time (ZWT),
Minimal-Waiting-Time (MWT), Bounded-Waiting-Time (BWT) and
Probabilistic-Waiting-Time (PWT) (the precise definitions are given
in Section~\ref{model}). Since the heavy traffic limits for the ZWT
and PWT classes can be directly derived from the current literature
(details in our technical report \cite{Yousi}), we simply list their
results, and focus instead on the MWT and BWT classes for which we
develop new heavy traffic limits. We use the heavy traffic limits to
characterize the relationship between the traffic intensity and the
number of servers in the queueing systems.

In applications, we show how to use these heavy traffic limit
results to determine the number of active machines in a cloud to
ensure that the QoS requirements are met and the cloud operates in a
stable and cost efficient manner. Cloud computing environments are
rapidly deployed by the industry as a means to provide efficient
computing resources. A significant fraction of the overall cost of
operating a cloud is the amount of power it consumes, which is
related to the number of machines in operation. In order to
efficiently manage the power cost associated with cloud computing,
we develop the foundations for designing a cloud computing
environment. In particular, we aim to determine how many machines a
cloud should have to sustain a specific system load and a certain
level of QoS, or equivalently how many machines should be kept awake
at any given time. Finally, using simulations we show that depending
on the QoS requirements of the cloud, the cloud needs substantially
different number of machines. We also show that the number of
operational machines in simulations are consistent with the proposed
design based on the new set of heavy traffic limit results. Although
the number of operational machines is derived from heavy traffic
limits, simulation results indicate that it is a good methodology,
even when the number of machines is finite, but large.


The main contributions of this paper are:

\begin{itemize}
\item
This paper makes new contributions to heavy traffic analysis, in
that it derives new heavy traffic limits for two important QoS
classes (MWT and BWT) for queueing systems when the arrival process
is general and the service times are hyper-exponentially
distributed.
\item
Using the heavy traffic limits results, this paper answers the
important question for enabling a power efficient cloud computing
environment as an application: How many machines should a cloud have
to sustain a specific system load and a certain level of QoS, or
equivalently how many machines should be kept awake at any given
time?
\end{itemize}

The paper is organized as follows. In Section~\ref{model}, we
present the system model of the queueing system, and describe the
four different classes of QoS requirements. Based on this model, we
develop heavy traffic limits results in
Section~\ref{htl_analysis_c2} and Section~\ref{htl_analysis_c3} for
the MWT and BWT classes correspondingly. Using these heavy traffic
limits results and the results in our technical report \cite{Yousi},
in Section~\ref{cloud} we consider cloud computing environment as an
application and compute the operational number of machines needed
for different classes of clouds. Simulation results are also
provided in Section~\ref{evaluation}. Finally, we conclude this
paper in Section~\ref{conclude}.

\section{System Model and QoS Classes}
\label{model}
\subsection{System Model and Preliminaries}


We assume that the queueing system consists of a large number of
servers, out of which $n$ are active/operational at any given time.
A larger $n$ will result in better QoS at the expense of higher
operational cost.

We assume that the job arrivals to the system are independent with
rate $\lambda_n$ and coefficient of variation $c$.

We also assume that the service time $v$ of the system satisfies the
hyper-exponential distribution as given below.

\begin{equation}
\label{H} P(v> t)=\sum_{i=1}^k P_i e^{-\mu_i t}
\end{equation}

Without loss of generality, we assume that

\begin{equation}
\label{condition} \begin{array}{c}
0<\mu_1<\mu_2<...<\mu_k<\infty;\\
P_i>0,\ \forall i\in{1,...k};\ \sum_{i=1}^{k}{P_i}=1.
\end{array}
\end{equation}

The maximum buffer size that holds the jobs that are yet to be
scheduled is assumed to be unbounded. The service priority obeys a
first-come-first-serve (FCFS) rule. In this paper we consider a
service model where each job is serviced by one server. All servers
are considered to have similar capability.


\subsection{Definition of QoS Classes}

Before we give the definition of different QoS classes, we first
provide some notations that will be used throughout this section.
Here, we let $n$ denote the total number of servers. For a given
$n$, we let $T_n$ denote the time that a job is in the system before
departure, $Q_n$ denote the total number of jobs in the system,
$W_n$ denote the time that the job waits in the system before being
processed. For two functions $f(n)$ and $g(n)$ of $n$,
$g(n)=o(f(n))$ if and only if $\lim\limits_{n
\rightarrow\infty}g(n)/f(n)=0$. Also, we use $\sim$ as equivalent
asymptotics, i.e., $f(n)\sim g(n)$ means that
$\lim\limits_{n\rightarrow\infty}f(n)/g(n)=1$. We also use
$\phi(\cdot)$ and $\Phi(\cdot)$ as probability density function and
cumulative distribution function of normal distribution, and use
$\varphi_X(\cdot)$ as the characteristic function of the random
variable $X$.

We now provide precise definitions of the various QoS classes
described in the introduction. Since we are interested in studying
the performance of the system in the heavy traffic limit, we let the
traffic intensity $\rho\rightarrow 1$ as $n \rightarrow \infty$ in
the case of each QoS class we study.

\subsubsection{Zero-Waiting-Time (ZWT) Class}
A system of the ZWT class is one for which

$$\lim\limits_{n\rightarrow\infty}P\{Q_n\geq n\}=0$$

The ZWT class corresponds to the class that provides the strictest
of the QoS requirements we consider here. For such systems, the
requirement is that an arriving job needs to wait in the queue is
zero. Loosely speaking, a system of the ZWT class corresponds to
having a QoS requirement that the jobs need to be served \emph{as
soon as} they arrive into the system.

\subsubsection{Minimal-Waiting-Time (MWT) Class}

For this class, the QoS requirement is
$$\lim\limits_{n\rightarrow\infty}P\{Q_n\geq n\}=\alpha,$$
where $\alpha$ is a constant such that $0<\alpha<1$.

This requirement is less strict than the ZWT class. There is a
nonvanishing probability that the jobs queue of the system is not
empty. Roughly speaking, a system of the MWT class corresponds to
the situation when jobs are served with some probability as soon as
they arrive into the system.

\subsubsection{Bounded-Waiting-Time (BWT) Class}
For this class,
$$\lim\limits_{n\rightarrow\infty}P\{Q_n\geq n\}=1$$
$${P\{W_n > t_1\}}\sim{\delta_n},$$
where $$\lim\limits_{n\rightarrow\infty}\delta_n=0.$$

The BWT class corresponds to the class for which the probability of
waiting time $W_n$ to exceed a constant threshold $t_1$ decreases to
0 as $n$ goes to infinity. The decreasing rate has equivalent
asymptotics with $\delta_n$. This means that the waiting time $W_n$
is between 0 and $t_1$ with probability 1, as $n$ goes to infinity.

\subsubsection{Probabilistic-Waiting-Time (PWT) Class}
For this class,
$$\lim\limits_{n\rightarrow\infty}P\{Q_n\geq n\}=1$$
$$\lim\limits_{n\rightarrow\infty}P\{W_n > t_2\}=\delta,$$
where $\delta$ is a given constant and satisfies $0<\delta<1$.

The PWT class corresponds to the class that provides the least
strict QoS requirements of the four types of systems considered
here. Hence, the probability that the waiting time $W_n$ is greater
than some constant threshold $t_2$ is non-zero, for large enough
$n$. This means that the QoS requirement for this system is such
that the waiting time $W_n$ is between 0 and $t_2$ with probability
$1-\delta$, as $n$ goes to infinity.

Further discussions and details on the four classes is given in
Section~\ref{evaluation} and our technical report \cite{Yousi}. For
the rest of the paper, we will mainly focus on developing new heavy
traffic limits for the MWT and BWT classes.

\section{Heavy Traffic Limit Analysis for the MWT class}

\label{htl_analysis_c2}

The following result tells us how the number of servers must scale
in the heavy traffic limit for the MWT class.

\begin{proposition}
\label{hyper}

Assume
\begin{equation}
\label{c2-rho-limit} \lim\limits_{n\rightarrow\infty}\rho_n=1,
\end{equation}

\begin{equation}
\label{p1} \lim\limits_{n\rightarrow\infty}P\{Q_n\geq n\}=\alpha,
\end{equation}

then

\begin{equation}
\label{ineq} L\leq\lim\limits_{n\rightarrow
\infty}(1-\rho_n)\sqrt{n}\leq U,
\end{equation}

where

\begin{equation}
\label{upper-bound}
U=\left(\sum_{i=1}^k{\beta_U^{(i)}\sqrt{\frac{P_i}{\mu_i}}}\right)\sqrt{\mu},
\end{equation}

\begin{equation}
\label{lower-bound}
L=\max_{i\in{\{1,...k\}}}\left\{{\beta_L^{(i)}\sqrt{\frac{P_i}{\mu_i}}}\right\}\sqrt{\mu},
\end{equation}

\begin{equation}
\label{def-rho}
\mu=\left(\sum_{i=0}^k\frac{P_i}{\mu_i}\right)^{-1},\quad
\rho_n=\frac{\lambda_n}{n\mu},
\end{equation}

\begin{equation}
\label{whitt-bound1}
\begin{array}{c}
\beta_U^{(i)}=(1+\frac{c^2-1}{2}P_i)\psi_U,\\
\frac{\alpha}{k}=[1+\sqrt{2\pi}\psi_U\Phi(\psi_U)\exp{(\psi_U^2/2)}]^{-1},\\
\end{array}
\end{equation}

\begin{equation}
\label{whitt-bound2}
\begin{array}{c}
\beta_L^{(i)}=(1+\frac{c^2-1}{2}P_i)\psi_L,\\
\alpha=[1+\sqrt{2\pi}\psi_L\Phi(\psi_L)\exp{(\psi_L^2/2)}]^{-1},\\
\end{array}
\end{equation}


\begin{equation}
\label{whitt-bound3}
\begin{array}{c}
0\leq\alpha\leq 1,\quad 0\leq\beta_L\leq \infty,\quad
0\leq\beta_U\leq
\infty.\\
\end{array}
\end{equation}

\end{proposition}

In Proposition~\ref{hyper}, $\psi_U$ is the solution of
Eq.~(\ref{whitt-bound1}), and $\beta_U^{(i)}$ can be computed using
$\psi_U$. Similarly, $\psi_L$ is the solution of
Eq.~(\ref{whitt-bound2}), and $\beta_L^{(i)}$ can be computed using
$\psi_L$. Thus, upper bound $U$ in Eq.~(\ref{upper-bound}) and lower
bound $L$ in Eq.~(\ref{lower-bound}) can be achieved using
$\beta_U^{(i)}$, $\beta_L^{(i)}$, and other parameters.

To prove Proposition~\ref{hyper}, we construct an \emph{artificial
system structure}. The arrival process and the capacity of a single
server are same as the original system. In the artificial system, we
assume that there are $k$ types of jobs. For each arrival, we know
the probability of $i^{th}$ type job is $P_i$, and the service time
of each $i^{th}$ type job is exponentially distributed with mean
$1/\mu_i$. Thus, the service time $v$ of the system can be viewed as
a hyper-exponential distribution which satisfies Eq.\ (\ref{H}). We
also assume that there is an omniscient scheduler for the artificial
system. This scheduler can recognize the type of arriving jobs, and
send them to the corresponding queue. For arrivals of type $i$, the
scheduler sends them to the $i^{th}$ queue, which contains $n_i$
servers. Then the arrival rate of the $i^{th}$ queue is $P_i
\lambda_n$. Also, the priority of each separated queue obeys the
FCFS rule. The artificial system is shown in Fig.~\ref{comp-sys}.

\begin{figure}
\centering
\includegraphics[width=0.6\textwidth]{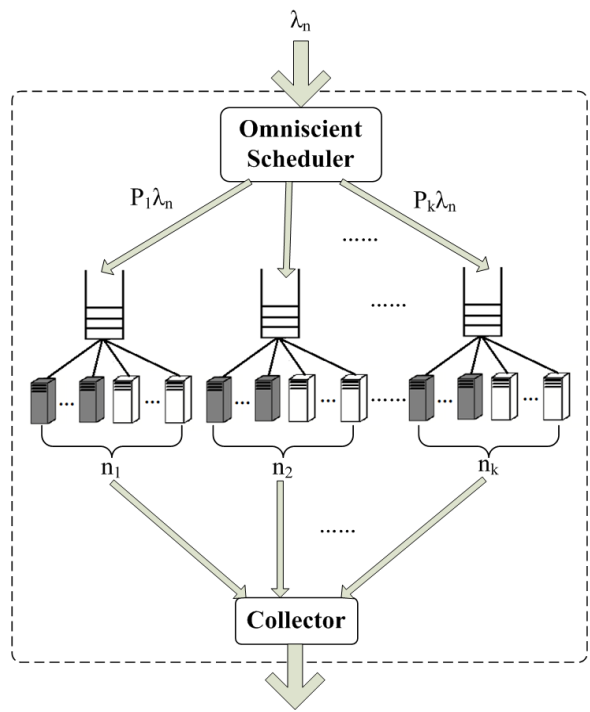}
\caption{Artificial System Structure} \label{comp-sys}
\end{figure}

\begin{lemma}
For the $i^{th}$ separated queue, the inter-arrival time
$\{Y_j^{(i)},\ j=1,2,...\}$ is i.i.d., and the coefficient of
variance $c^{(i)}=\sqrt{1+(c^2-1)P_i}$.
\end{lemma}
\begin{proof}
For the $i^{th}$ separated queue in Fig.~\ref{comp-sys}, the
inter-arrival time $Y^{(i)}$ is a summation of inter-arrival times
of a certain number of consecutive arrivals in the original queue.
The number of the summands is a random variable $k_j^{(i)}$.
$k_j^{(i)}$ is equal to the number of original arrivals between
$(j-1)^{th}$ and $j^{th}$ arrivals in the $i^{th}$ separated queue.



Based on the structure of the artificial system, $k_j^{(i)}$ is an
independent random variable with geometric distribution with
parameter $P_i$. Assume $\{X_1,X_2,...\}$ are the inter-arrival
times in the original queueing system. Note that $\{X_1,X_2,...\}$
are also independent of $k_j^{(i)}$, because $k_j^{(i)}$ is only
dependent on the distribution of the service time. Then, for each
$i$, the inter-arrival time $\{Y_j^{(i)},\ j=1,2,...\}$ is i.i.d..

Let $t$ be the index of the first inter-arrival time within the
$j^{th}$ inter-arrival time in separated queue $i$. Then,
$Y_j^{(i)}=X_t+X_{t+1}+...+X_{t+k_j^{(i)}-1}$. So,
\begin{equation}
\begin{split}
&E(Y_j^{(i)})=E(X_t+X_{t+1}+...+X_{t+k_j^{(i)}-1})\\
=&E(E(X_t+X_{t+1}+...+X_{t+k_j^{(i)}-1}|k_j^{(i)}))\\
=&E(k_j^{(i)} E(X_t))=E(k_j^{(i)})E(X_t),
\end{split}
\end{equation}
and
\begin{equation}
\begin{split}
&Var(Y_j^{(i)})=E\left((Y_j^{(i)})^2\right)-\left(E(Y_j^{(i)})\right)^2\\
=&E\left((X_t+X_{t+1}+...+X_{t+k_j^{(i)}-1})^2\right)-\left(E(Y_j^{(i)})\right)^2\\
=&E(E\left((X_t+X_{t+1}+...+X_{t+k_j^{(i)}-1})^2|k_j^{(i)}\right))-\left(E(Y_j^{(i)})\right)^2\\
=&E(E(X_t^2+X_{t+1}^2+...+X_{t+k_j^{(i)}-1}^2+\\
&2 X_t X_{t+1}+...+2X_{t+k_j^{(i)}-2}X_{t+k_j^{(i)}-1}|k_j^{(i)}))-\left(E(Y_j^{(i)})\right)^2\\
=&E\left((k_j^{(i)})^2(E(X_t))^2+k_j^{(i)} Var(X_t)\right)-\left(E(Y_j^{(i)})\right)^2\\
=&E((k_j^{(i)})^2)(E(X_t))^2+E(k_j^{(i)}) Var(X_t)-E(k_j^{(i)})^2 E(X_t)^2\\
=&Var(k_j^{(i)})E(X_t)^2+E(k_j^{(i)})Var(X_t).
\end{split}
\end{equation}
Thus, we can achieve the coefficient of variation $c^{(i)}$ for all
the separated queues as below.

\begin{equation}
\begin{split}
&c^{(i)}=\sqrt{Var(Y_j^{(i)})/\left(E(Y_j^{(i)})\right)^2}\\
=&\sqrt{\frac{Var(k_i)\left(E(X_t)\right)^2+E(k_i)Var(X_t)}{E(k_i)^2E(X_t)^2}}\\
=&\sqrt{\frac{\frac{1-P_i}{P_i^2}E(X_t)^2+\frac{1}{P_i}Var(X_t)}{\frac{1}{P_i^2}E(X_t)^2}}=\sqrt{1+(c^2-1)P_i}.
\end{split}
\end{equation}

\end{proof}

\begin{remark}
If the arrival process is Poisson, $c=1$, then $c^{(i)}=1,\ \forall
i=1,2,...k$. If the arrival process is deterministic, $c=0$, then
the inter-arrival time of each separated queue has a geometric
distribution, and $c^{(i)}=\sqrt{1-P_i},\ \forall i=1,2,...k$.
\end{remark}

\begin{proof}[Proof of Proposition~\ref{hyper}]
To prove this proposition, we must prove both the upper and the
lower bounds of the limit. For the upper bound, we consider the
Artificial System \uppercase\expandafter{\romannumeral1}, which
satisfies the following condition:

\begin{equation}
\label{c2_lim}
\lim\limits_{n_i\rightarrow\infty}(1-\rho_{n_i})\sqrt{n_i}=\beta_U^{(i)},
\end{equation}
where
\begin{equation}
\begin{split}
\label{c2_lim_where} \rho_{n_i}&=\frac{P_i \lambda_n}{n_i
\mu_i},\\
\quad
\beta_U^{(i)}&=\frac{(1+(c^{(i)})^2)\psi_U}{2}=(1+\frac{c^2-1}{2}P_i)\psi_U,
\end{split}
\end{equation}
and
\begin{equation}
\frac{\alpha}{k}=[1+\sqrt{2\pi}\psi_U\Phi(\psi_U)\exp{(\psi_U^2/2)}]^{-1}.
\end{equation}

The result of Theorem 4 in \cite{whitt1981} shows that
\begin{equation}
\label{whitt1} \lim\limits_{n\rightarrow \infty}P\{Q_n\geq
n\}=\alpha_c
\end{equation}
if and only if
\begin{equation}
\label{whitt2} \lim\limits_{n\rightarrow\infty}(1-\rho_n)\sqrt{n}=\beta,\\
\end{equation}
under the following conditions:
\begin{equation}
\label{whitt3}\begin{array}{c}
\beta=\frac{(1+c^2)\psi}{2},\\
\alpha_c=[1+\sqrt{2\pi}\psi\Phi(\psi)\exp{(\psi^2/2)}]^{-1},\\
0\leq\alpha_c\leq 1,\quad 0\leq\beta\leq \infty.
\end{array}
\end{equation}

By applying this result into Artificial System
\uppercase\expandafter{\romannumeral1}, for each individual queue,
we have

\begin{equation}
\lim\limits_{n_i\rightarrow \infty}P\{Q^{(i)}_{n_i}\geq
n_i\}=\frac{1}{k}, \quad \forall i\in \{1,...k\},
\end{equation}

where $Q^{(i)}_{n_i}$ is the length of the $i^{th}$ separated queue.

Let $n_U=\sum_{i=1}^k {n_i}$, $Q_{n_U}=\sum_{i=1}^k{Q^{(i)}_{n_i}}$.
Then, for Artificial System \uppercase\expandafter{\romannumeral1},
we have

\begin{equation}
\label{pf-ub-relax}
\begin{split}
&P\{Q_{n_U}\geq n_U\}=P\{\sum_{i=1}^k Q_{n_U}^{(i)} \ge \sum_{i=1}^k n_i\}\\
\leq & P\left(\bigcup_{i=1}^k \{Q^{(i)}_{n_i}\ge n_i\}\right) \leq
\sum_{i=1}^k P\{Q^{(i)}_{n_i}\ge n_i\} .
\end{split}
\end{equation}

By taking the limit on both sides,
\begin{equation}
\label{pf-ub}
\begin{split}
&\lim\limits_{\begin{subarray}{c} n_i\rightarrow \infty\\ i\in
\{1,...k\}\end{subarray}}P\{Q_{n_U}\geq n_U\}\\
\leq &\lim\limits_{\begin{subarray}{c} n_i\rightarrow \infty\\ i\in
\{1,...k\}\end{subarray}}\left(\sum_{i=1}^k P\{Q^{(i)}_{n_i}\ge n_i\}\right)\\
= & \left(\sum_{i=1}^k \lim\limits_{n_i\rightarrow \infty}
P\{Q^{(i)}_{n_i}\ge n_i\}\right)=\alpha
\end{split}
\end{equation}

From Eq. (\ref{pf-ub}), we know that when Artificial System \uppercase\expandafter{\romannumeral1} has $n_U$ servers, the probability that queue length
$Q_{n_U}$ is greater than or equal to $n_U$ is asymptotically less than or equal to $\alpha$. Observe that the original system needs no more servers than
Artificial System \uppercase\expandafter{\romannumeral1} since there may be some idle servers in Artificial System \uppercase\expandafter{\romannumeral1},
even when the other job queues are not empty. Based on the asymptotic optimality of FCFS in our system
\cite{schedule::tail,schedule::possible,schedule::control,schedule::large}, to satisfy the same requirement, the original system does not need more servers
than Artificial System \uppercase\expandafter{\romannumeral1}. By using Eqs. (\ref{c2_lim}) and (\ref{c2_lim_where}), we can solve for $n_i$. That is,

\begin{equation}
\label{pf-ub-2}
\begin{split}
n\leq n_U&=\sum_{i=1}^k{n_i}=\sum_{i=1}^k{\left(\frac{P_i
\lambda_n}{\mu_i}+\beta_U^{(i)} \sqrt{\frac{P_i
\lambda_n}{\mu_i}}\right)}\\
&=\frac{\lambda_n}{\mu}+\sqrt{\frac{\lambda_n}{\mu}}\left(\sum_{i=1}^k{\beta_U^{(i)}\sqrt{\frac{P_i}{\mu_i}}}\right)\sqrt{\mu}.
\end{split}
\end{equation}

Since $\lim\limits_{n_i\rightarrow\infty}\rho_i=1$, we ignore the
factor $\sqrt{\frac{1}{\rho_i}}$ and achieve Eq.~(\ref{pf-ub-2}). By
taking Eq. (\ref{pf-ub-2}) into the definition of $\rho_n$ in Eq.
(\ref{def-rho}), we can directly achieve the upper bound Eq.
(\ref{upper-bound}) of Eq. (\ref{ineq}).

For the lower bound, we consider Artificial System
\uppercase\expandafter{\romannumeral2}, which has similar structure
as Artificial System \uppercase\expandafter{\romannumeral1} and
Fig.~\ref{comp-sys}, but $n_i$ satisfies the following conditions.

\begin{equation}
\label{lb-ni} n_i=\left\{\begin{array}{ll}
\frac{P_i \lambda_n}{\mu_i}, &i \in \{1,...k\},\ i\neq m\\
\frac{P_m \lambda_n}{\mu_m}+\beta_L^{(m)}\sqrt{\frac{P_m
\lambda_n}{\mu_m}}, &i=m
\end{array}\right.
\end{equation}

where
\begin{equation}
\label{lb-ni2}
\begin{split}
\beta_L^{(i)}&=\frac{(1+(c^{(i)})^2)\psi}{2}=(1+\frac{c^2-1}{2}P_i)\psi_L,\\
m&=\inf\underset{i\in\{1,...k\}}{\mathrm{argmax}}\left(\beta_L^{(i)}\sqrt{\frac{P_i}{\mu_i}}\right),
\end{split}
\end{equation}
and
\begin{equation}
\label{lb-ni3}
\alpha=[1+\sqrt{2\pi}\psi_L\Phi(\psi_L)\exp{(\psi_L^2/2)}]^{-1}.
\end{equation}

Then,
\begin{equation}
\lim\limits_{n_m\rightarrow\infty}(1-\rho_{n_m})\sqrt{n_m}=\beta_L^{(m)},
\end{equation}

where

\begin{equation}
\rho_{n_m}=\frac{P_m \lambda_n}{n_m \mu}.
\end{equation}

By substituting Eqs. (\ref{whitt1}-\ref{whitt3}) into Eqs.
(\ref{lb-ni}-\ref{lb-ni3}), the reader can verify the following
result for Artificial System \uppercase\expandafter{\romannumeral2}.

\begin{equation}
\label{lower-prob} \lim\limits_{n_i\rightarrow
\infty}P\{Q^{(i)}_{n_i}\geq n_i\}=\left\{
\begin{array}{ll}
1, &i\in\{1,...k\},\ i\neq m\\
\alpha, &i=m
\end{array}
\right.
\end{equation}

Define $n_L=\sum_{i=1}^k{n_i}$. If the original system has $n_L$ servers, then we can construct a scheduler based on Artificial System
\uppercase\expandafter{\romannumeral2}. This scheduler can make QoS of the arrivals satisfy Eq. (\ref{lower-prob}). By the effect of the scheduler, this
queueing discipline is neither FCFS nor work conserving. The original system, needs more servers than Artificial System
\uppercase\expandafter{\romannumeral2} to satisfy Eq. (\ref{p1}) (see details in our technical report \cite{Yousi}). Therefore, $n$ should be greater than
or equal to $n_L$, i.e.,

\begin{equation}
\label{pf-lb}
\begin{split}
n\geq n_L&=\sum_{i=1}^k{n_i}=\sum_{i=1}^{k}{\left(\frac{P_i
\lambda_n}{\mu_i}\right)}+\beta_L^{(m)} \sqrt{\frac{P_m
\lambda_n}{\mu_m}}\\
&=\frac{\lambda_n}{\mu}+\sqrt{\frac{\lambda_n}{\mu}}\max_{i\in{\{1,...k\}}}\left\{\beta_L^{(i)}\sqrt{\frac{P_i}{\mu_i}}\right\}\sqrt{\mu}.
\end{split}
\end{equation}

By taking Eq. (\ref{pf-lb}) into the definition of $\rho_n$ in Eq.
(\ref{def-rho}), we can directly achieve the lower bound Eq.
(\ref{lower-bound}) of Eq. (\ref{ineq}).

\end{proof}

\begin{corollary}
If the arrival process is Poisson process, we have a tighter upper
bound $\widehat{U}$, which satisfies the following equation.
\begin{equation}
\label{poisson-upper-bound}
\widehat{U}=\left(\sum_{i=1}^k{\sqrt{\frac{P_i}{\mu_i}}}\right)\sqrt{\mu}\widehat{\psi_U},
\end{equation}
where
\begin{equation}
\label{poisson-def-rho}
\mu=\left(\sum_{i=0}^k\frac{P_i}{\mu_i}\right)^{-1},\quad
\rho_n=\frac{\lambda_n}{n\mu},
\end{equation}

\begin{equation}
\label{poisson-whitt-bound1}
1-(1-\alpha)^{\frac{1}{k}}=[1+\sqrt{2\pi}\widehat{\psi_U}\Phi(\widehat{\psi_U})\exp{(\widehat{\psi_U}^2/2)}]^{-1},\\
\end{equation}

\begin{equation}
\label{poisson-whitt-bound3} 0\leq\alpha\leq 1,\quad
0\leq\widehat{\psi_U}\leq \infty.
\end{equation}

\end{corollary}
\begin{proof}

For Poisson arrival process, we can easily achieve that $c=1$ and
$c^{(i)}=1,\ \forall i\in\{1,2,...,k\}$. We consider a similar
Artificial System \uppercase\expandafter{\romannumeral3}, which has
same structure as Artificial System
\uppercase\expandafter{\romannumeral2}. Let Artificial System
\uppercase\expandafter{\romannumeral3} satisfy the following
conditions.

\begin{equation}
\label{poisson_c2_lim}
\lim\limits_{n_i\rightarrow\infty}(1-\rho_{n_i})\sqrt{n_i}=\widehat{\psi_U},
\end{equation}
where
\begin{equation}
\label{poisson_c2_lim_where} \rho_{n_i}=\frac{P_i \lambda_n}{n_i
\mu_i},
\end{equation}
and
\begin{equation}
1-(1-\alpha)^{\frac{1}{k}}=[1+\sqrt{2\pi}\widehat{\psi_U}\Phi(\widehat{\psi_U})\exp{(\widehat{\psi_U}^2/2)}]^{-1}.
\end{equation}

Similarly to Artificial System
\uppercase\expandafter{\romannumeral2}, for each individual queue,
we have

\begin{equation}
\lim\limits_{n_i\rightarrow \infty}P\{Q^{(i)}_{n_i}\geq
n_i\}=1-(1-\alpha)^{\frac{1}{k}}, \quad \forall i\in \{1,...k\},
\end{equation}

where $Q^{(i)}_{n_i}$ is the length of the $i^{th}$ separated queue.

Let $n_U=\sum_{i=1}^k {n_i}$. Since arrival process is Poisson
process, by the Colouring Theorem \cite{kingman_book}, the arrival
process in each separated queue is independent Poisson process.
Then, for Artificial System \uppercase\expandafter{\romannumeral3},
we have

\begin{equation}
\begin{split}
&P\{Q_{n_U}\geq n_U\}=1-P\{Q_{n_U} < n_U\}\\
\leq & 1-\prod_{i=1}^k{\left(1-P\{Q^{(i)}_{n_i}\geq n_i\}\right)}
\end{split}
\end{equation}
where
\begin{equation}
Q_{n_U}=\sum_{i=1}^k{Q^{(i)}_{n_i}}.
\end{equation}

By taking the limits on each sides, we can achieve that
\begin{equation}
\label{poisson-pf-ub}
\begin{split}
&\lim\limits_{\begin{subarray}{c} n_i\rightarrow \infty\\ i\in
\{1,...k\}\end{subarray}}P\{Q_{n_U}\geq n_U\}\\
\leq &\lim\limits_{\begin{subarray}{c} n_i\rightarrow \infty\\ i\in
\{1,...k\}\end{subarray}}\left(1-\prod_{i=1}^k{\left(1-P\{Q^{(i)}_{n_i}\geq
n_i\}\right)}\right)\\
= & 1-\prod_{i=1}^k{\left(1-\lim\limits_{n_i\rightarrow
\infty}P\{Q^{(i)}_{n_i}\geq n_i\}\right)} =\alpha
\end{split}
\end{equation}

From Eq. (\ref{poisson-pf-ub}), we know that when artificial system
\uppercase\expandafter{\romannumeral1} has $n_U$ servers, the
probability that queue length $Q_{n_U}$ is greater than or equal to
$n_U$ is asymptotically less than or equal to $\alpha$. To satisfy
the same requirement, the original system does not need more servers
than Artificial System \uppercase\expandafter{\romannumeral3}. By
using Eqs. (\ref{poisson_c2_lim}) and (\ref{poisson_c2_lim_where}),
we can get the expression of $n_i$. That is,

\begin{equation}
\label{poisson-pf-ub-2}
\begin{split}
n\leq n_U&=\sum_{i=1}^k{n_i}=\sum_{i=1}^k{\left(\frac{P_i
\lambda_n}{\mu_i}+\widehat{\psi_U} \sqrt{\frac{P_i
\lambda_n}{\mu_i}}\right)}\\
&=\frac{\lambda_n}{\mu}+\sqrt{\frac{\lambda_n}{\mu}}\left(\sum_{i=1}^k{\sqrt{\frac{P_i}{\mu_i}}}\right)\sqrt{\mu}\widehat{\psi_U}.
\end{split}
\end{equation}

By taking Eq. (\ref{poisson-pf-ub-2}) into the definition of
$\rho_n$ in Eq. (\ref{poisson-def-rho}), we can directly achieve the
upper bound Eq. (\ref{poisson-upper-bound}).

Since for Poisson arrival process, $c=1$ and $c^{(i)}=1,\ \forall
i\in\{1,2,...,k\}$, then $\beta_U^{(i)}=\psi_U$ in
Eq.(\ref{whitt-bound1}). Since $(1-\frac{\alpha}{k})^k$ is an
increasing function, then $(1-\frac{\alpha}{k})^k \ge 1-\alpha$.
Thus, $1-(1-\alpha)^{\frac{1}{k}} \ge \frac{\alpha}{k}$. We can
directly achieve that $\psi_U \ge \widehat{\psi_U}$, i.e.,
Eq.(\ref{poisson-upper-bound}) is a tighter upper bound then
Eq.(\ref{upper-bound}) for Poisson arrival process.

\end{proof}


\begin{remark}
When $k=1$, the service time reduces to an exponential distribution.
Based on the Proposition~\ref{hyper}, we can see that
$U=L=\beta_U=\beta_L\triangleq\beta$ in this scenario, i.e.,
$\lim\limits_{n\rightarrow \infty}(1-\rho_n)\sqrt{n}=\beta$. Thus,
Proposition~\ref{hyper} in our paper is consistent with Proposition
1 and Theorem 4 in \cite{whitt1981}.
\end{remark}

\begin{corollary}
\label{tight-ub-new} The solution $\widetilde{U}$ of the following
optimization problem results in a tighter upper bound for the Eq.
(\ref{ineq}).

\begin{equation}
\label{op-new} \min\limits_{\alpha_1,...,\alpha_k}\
\frac{\sum_{j=1}^k {\beta_j \sqrt{\frac{P_j}{\mu_j}}}
}{\sqrt{\sum_{j=1}^k \frac{P_j}{\mu_j}}},
\end{equation}

\begin{equation}
\label{op2-new} \mathbf{s.t.}\ \sum_{j=1}^k \alpha_j \leq \alpha,
\end{equation}
where
\begin{equation}
\begin{array}{c}
\beta_j=(1+\frac{c^2-1}{2}P_j)\psi_j,\\
\alpha_j=[1+\sqrt{2\pi}\psi_j\Phi(\psi_j)\exp{(\psi_j^2/2)}]^{-1},\\
0\leq\alpha_j\leq 1,\quad 0\leq\beta_j\leq \infty,\quad \forall j.
\end{array}
\end{equation}

%

\end{corollary}

\begin{proof}

It is not necessary to choose all the $\alpha_j$ equally. Once
Eq.~\ref{pf-ub-relax} is satisfied, it is sufficient to find an
upper bound. Thus, the minimum of all the upper bounds are a new
tighter upper bound for Proposition~\ref{hyper}.
\end{proof}

\begin{remark}
Since the corresponding objective value of every $\{\alpha_j,\
j=1,...,k\}$ in the feasible set of the optimization problem
(\ref{op-new}-\ref{op2-new}) is an upper bound of the limit in
(\ref{ineq}). If we choose $\widetilde{\alpha_j}=\frac{\alpha}{k},\
\forall j=1,...,k$, it is easy to check that the value of
$\{\widetilde{\alpha_j},\ j=1,...,k\}$ is in the feasible set, and
the objective value is same as the upper bound in Eq.
(\ref{upper-bound}).
\end{remark}

\begin{corollary}
\label{tight-ub} The solution $\widetilde{U}$ of the following
optimization problem results in a tighter upper bound for Poisson
arrival process.

\begin{equation}
\label{op} \min\limits_{\alpha_1,...,\alpha_k}\ \frac{\sum_{j=1}^k
{\beta_j \sqrt{\frac{P_j}{\mu_j}}} }{\sqrt{\sum_{j=1}^k
\frac{P_j}{\mu_j}}},
\end{equation}

\begin{equation}
\label{op2} \mathbf{s.t.}\
\lim_{s\rightarrow\infty}\int_{-\infty}^{\infty}
\left[\prod_{j=1}^k\varphi_{\widehat{Q_j}}\left(\sqrt{\frac{P_i}{\mu_i}}t\right)\frac{1-\exp{(-its)}}{it}\right]dt
\leq 2\pi\alpha,
\end{equation}
where
\begin{equation}
\begin{array}{c}
\beta_j=\frac{(1+c^2)\psi_j}{2},\\
\alpha_j=[1+\sqrt{2\pi}\psi_j\Phi(\psi_j)\exp{(\psi_j^2/2)}]^{-1},\\
0\leq\alpha_j\leq 1,\quad 0\leq\beta_j\leq \infty,\quad \forall j,
\end{array}
\end{equation}

and the probability density function of $\widehat{Q_j}$ is

\begin{equation}
\label{distribution} f_j(x)=\left\{
\begin{array}{lc}
\alpha_j\beta_j \exp{(-\beta_j x)}, & when\ x>0\\
(1-\alpha_j)\frac{\phi(x+\beta_j)}{\Phi(\beta_j)}, & when\ x<0
\end{array}
 \right. .
\end{equation}

\end{corollary}

\begin{proof}
\label{app-tub} We construct a new comparable system with similar
structure as Fig.~\ref{comp-sys}. For sub-queue $j$, let the
probability that queue length $Q_j$ is greater than or equal to
$n_j$ be $\alpha_j$. Then, the total number of servers $n$ is

\begin{equation}
\label{arti-n}
\begin{split}
n&=\sum_{j=1}^k n_j=\left(\sum_{j=1}^k \frac{P_j}{\mu_j}\right)
\lambda+\left(\sum_{j=1}^k
\beta_j\sqrt{\frac{P_j}{\mu_j}}\right)\sqrt{\lambda}\\
=&\frac{\lambda}{\mu}+\frac{\sum_{j=1}^k {\beta_j
\sqrt{\frac{P_j}{\mu_j}}} }{\sqrt{\sum_{j=1}^k \frac{P_j}{\mu_j}}}
\sqrt{\frac{\lambda}{\mu}},
\end{split}
\end{equation}
where $\mu$ is same as Eq. (\ref{def-rho}).

For each arrival, the end-to-end time $D$ of the original system is less than or equal to the end-to-end time $\widetilde{D}$ of the compared separated
system in stochastic ordering \cite{schedule::tail,daley,whitt84,towsley}. Then, there exists a sample space $\Omega$, such that
$D(\omega)\leq\widetilde{D}(\omega)$ \cite{fcfs,sp}. In this sample space $\Omega$, the queue length $Q(\omega)$ of the original system is less than or
equal to the total queue length $\widetilde{Q}(\omega)$ of the compared artificial system for all $\omega\in\Omega$. Thus, $Q\leq\widetilde{Q}$ in the
stochastic ordering. We represent this stochastic ordering as $Q\leq_{st}\widetilde{Q}$.

By the definition of the stochastic ordering \cite{sp}, for the same
number $n$, $P(\widetilde{Q} \geq n)\geq P(Q\geq n)$. In other
words, if we assume that the QoS of the artificial system can
satisfy $P(\widetilde{Q}\geq n)\leq\alpha$, then, to achieve the
same QoS, the original system needs no more than $n$ servers. For
this reason, we can achieve a tighter upper bound for Eq.
(\ref{ineq}).

Now, consider the artificial system with the same QoS. We define
$\widehat{Q_j}$ as $\frac{\widetilde{Q_j}-n_j}{\sqrt{n_j}}$. Then,
\begin{equation}
\label{arti-ineq}
\begin{split}
&\alpha\geq P\left(\sum_{j=1}^k \widetilde{Q_j}\geq
n\right)=P\left(\sum_{j=1}^k
(n_j+\sqrt{n_j}\widehat{Q_j})\geq n\right)\\
=&P\left(\sum_{j=1}^k \sqrt{n_j}\widehat{Q_j}\geq
0\right)=P\left(\sum_{j=1}^k
\sqrt{\frac{P_j}{\mu_j}}\widehat{Q_j}\geq 0\right)
\end{split}
\end{equation}

From Theorems 1 and 4 in \cite{whitt1981}, we can achieve the
probability of normalized queue length as Eq. (\ref{distribution}).
Then, the characteristic function of $\sum_{j=1}^k
\sqrt{\frac{P_j}{\mu_j}}\widehat{Q_j}$ in Eq. (\ref{arti-ineq}) is

\begin{equation}
\label{ch} \varphi_{\sum_{j=1}^k
\sqrt{\frac{P_j}{\mu_j}}\widehat{Q_j}} (t)=\prod_{j=1}^k
\varphi_{\sqrt{\frac{P_j}{\mu_j}}\widehat{Q_j}} (t)=\prod_{j=1}^k
\varphi_{\widehat{Q_j}} (\sqrt{\frac{P_j}{\mu_j}}t).
\end{equation}

By Levy's inversion theorem \cite{pnm}, the Eq. (\ref{arti-ineq})
can be written as

\begin{equation}
\label{arti-con}
\begin{split}
\alpha &\geq P\left(\sum_{j=1}^k
\sqrt{\frac{P_j}{\mu_j}}\widehat{Q_j}\geq 0\right)\\
&\geq \frac{1}{2\pi}
\lim_{s\rightarrow\infty}\int_{-\infty}^{\infty}
\left[\prod_{j=1}^k\varphi_{\widehat{Q_j}}\left(\sqrt{\frac{P_i}{\mu_i}}t\right)\frac{1-\exp{(-its)}}{it}\right]dt
\end{split}
\end{equation}

Thus, from Eq. (\ref{arti-n}) and (\ref{arti-con}), the solution of
optimization problem (\ref{op}-\ref{op2}) is an upper bound of the
limit in Eq. (\ref{ineq}) for the artificial system. Then, for the
original system, no more servers are needed under the same value of
traffic intensity, i.e., the upper bound of the artificial system is
also an upper bound for the original system.
\end{proof}

\begin{remark}
If we choose any $\{\alpha_j,\ j=1,...,k\}$ in the feasible set of
the optimization problem (\ref{op}-\ref{op2}), then the
corresponding objective value is an upper bound for Poisson
arrivals. If we choose
$\widetilde{\alpha_j}=1-(1-\alpha)^{\frac{1}{k}},\ \forall
j=1,...,k$, it is easy to check that the value of
$\{\widetilde{\alpha_j},\ j=1,...,k\}$ is in the feasible set, and
the objective value is same as the upper bound in
Eq.~(\ref{poisson-upper-bound}).
\end{remark}

\section{Heavy Traffic Limit Analysis for the BWT Class}
\label{htl_analysis_c3}

The following result provides conditions under which the waiting
time of a job is bounded by a constant $t_1$ but the probability
that new arrivals need to wait approaches one in the heavy traffic
scenario.

\begin{proposition}
\label{heavytraffic} Assume
\begin{equation}
\label{delta-limit} \lim\limits_{n\rightarrow\infty}\delta_n=0,
\end{equation}
then
\begin{equation}
\label{rho-limit} \lim\limits_{n\rightarrow\infty}\rho_n=1
\end{equation}

\begin{equation}
\label{c3p1} \lim\limits_{n\rightarrow\infty}P\{Q_n\geq n\}=1
\end{equation}

\begin{equation}
\label{p2} P\{W_n > t_1\}\sim\delta_n
\end{equation}

if and only if

\begin{equation}
\label{rho-real}
\lim\limits_{n\rightarrow\infty}\frac{(1-\rho_n)n}{-\ln{\delta_n}}=\tau
\end{equation}

\begin{equation}
\label{rho-sqrt-n} \lim\limits_{n\rightarrow \infty}\delta_n
\exp{(k\sqrt{n})}=\infty,\quad \forall k>0
\end{equation}

where
\begin{equation}
\label{def-tau} \tau=\frac{\mu^2\sigma^2+c^2}{2\mu t_1},\
\rho_n=\frac{\lambda_n}{n\mu},
\end{equation}

\begin{equation}
\label{def-sigma}
\mu=\left(\sum_{i=1}^k{\frac{P_i}{\mu_i}}\right)^{-1},\
\sigma^2=2\sum_{i=1}^k{\left(\frac{P_i}{\mu_i^2}\right)}-\left(\sum_{i=1}^k{\frac{P_i}{\mu_i}}\right)^2.
\end{equation}
\end{proposition}

\begin{remark}
The main reason why Proposition~\ref{heavytraffic} can be derived
from Proposition~\ref{hyper} is due to the asymptotic rate of
$\rho_n$. Although
$\lim\limits_{n\rightarrow\infty}(1-\rho_n)\sqrt{n}$ is no longer a
constant, it still has a constant lower and upper bound, i.e., it is
still on a constant ``level''.
\end{remark}

\begin{proof}[Proof of Proposition~\ref{heavytraffic}]
To prove Proposition~\ref{heavytraffic}, we must prove both
necessary and sufficient conditions.

\noindent\textbf{Necessary Condition:} From the heavy traffic results given by Kingman \cite{kingman1965}
and Kollerstrom \cite{koll1,koll2}, the equilibrium waiting time in our system can be shown to asymptotically
follow an exponential distribution with parameter

\begin{equation}
\label{wn-app}
\frac{2(E(v_n)-\frac{E(s_n)}{n})}{Var(\frac{s_n}{n})+Var(v_n)}.
\end{equation}

In Eq. (\ref{wn-app}), $s_n$ is the service time, and $v_n$ is the
inter-arrival time. Assume the mean and variance of service time is
$\mu^{-1}$ and $\sigma^2$. Then, we get

\begin{equation}
\label{kingman} \begin{split} &P(W_n\geq t_1)\sim\\
&\exp{\left(-\frac{2(\frac{1}{\lambda_n}-\frac{1}{n\mu})}{\frac{\sigma^2}{n^2}+\frac{c_n^2}{\lambda_n^2}}t_1\right)}=\exp{\left(-\frac{2\mu(1-\rho_n)n}{\mu^2\sigma^2+c_n^2}t_1
\right)}
\end{split}
\end{equation}

Since $c_n=c$ and for this class the equilibrium waiting time
satisfies that $P(W_n\geq t_1)\sim \delta_n$, it implies that

\begin{equation}
\label{gen_class3}
\lim\limits_{n\rightarrow\infty}\frac{(1-\rho_n)n}{-\ln{\delta_n}}=\tau,
\end{equation}
where
$$\tau\triangleq\frac{\mu^2\sigma^2+c^2}{2\mu t_1},$$
$$\mu=\left(\sum_{i=1}^k{\frac{P_i}{\mu_i}}\right)^{-1},\ \sigma^2=2\sum_{i=1}^k{\left(\frac{P_i}{\mu_i^2}\right)}-\left(\sum_{i=1}^k{\frac{P_i}{\mu_i}}\right)^2.$$

Based on Proposition~\ref{hyper}, from $\lim\limits_{n\rightarrow
\infty}P\{Q_n\geq n\}=1$, we can achieve that
$\lim\limits_{n\rightarrow \infty}(1-\rho_n)\sqrt{n}=0$, i.e.,
$\lim\limits_{n\rightarrow
\infty}\frac{\ln{\frac{1}{\delta_n}}}{\sqrt{n}}=0$. This means that
$\ln{\frac{1}{\delta_n}}=o(\sqrt{n})$. Hence,
$\lim\limits_{n\rightarrow \infty}\delta_n
\exp{(k\sqrt{n})}=\infty,\quad \forall k>0$. Thus, Eq.\
(\ref{rho-sqrt-n}) is achieved.

\noindent\textbf{Sufficient Condition:} When Eq.\ (\ref{rho-sqrt-n})
is satisfied, we get $\ln{\frac{1}{\delta_n}}=o(n)$, i.e.,
$\lim\limits_{n\rightarrow
\infty}{\frac{\ln{\frac{1}{\delta_n}}}{n}}=0$, which is equivalent
to $\lim\limits_{n\rightarrow\infty}\rho_n=1$ based on Eq.\
(\ref{rho-real}). Hence, Eq.\ (\ref{rho-limit}) is achieved.

Now, based on Eqs.\ (\ref{rho-sqrt-n}) and (\ref{rho-limit}), and
using the heavy traffic limit result Eqs.\
(\ref{whitt-bound1})-(\ref{whitt-bound3}), the lower bound in
Proposition~\ref{hyper} should satisfy that
\begin{equation}
\label{L0} L=0.
\end{equation}

By applying Eq. (\ref{L0}) in Eq.
(\ref{lower-bound}-\ref{whitt-bound3}), we can directly obtain
$\lim\limits_{n\rightarrow \infty}P\{Q_n\geq n\}=1$. Hence, Eq.\
(\ref{c3p1}) is satisfied.

Based on Eq.\ (\ref{rho-real}), it can be shown that
\begin{equation}
\label{mid} \lim\limits_{n\rightarrow\infty}\frac{\exp{[-n
(1-\rho_n)/\tau]}}{\delta_n}=1.
\end{equation}
Based on Eq.\ (\ref{kingman}), we get
$\lim\limits_{n\rightarrow\infty}\frac{P\{W_n > t_1\}}{\delta_n}=1$.
That is $P\{W_n > t_1\}\sim\delta_n$. Eq.\ (\ref{p2}) is achieved.
\end{proof}

\begin{remark}
\label{comp-remark} Let $k=1$, then $\mu_1=\mu$ and $P_1=1$. We can
directly achieve the scenario with exponential distributed service
time from Proposition~\ref{heavytraffic}. In the case of exponential
distributed service time, the Proposition~\ref{heavytraffic} still
holds, and $\tau$ can be simplified to $\frac{1+c^2}{2\mu t_1}$.
\end{remark}

\begin{corollary}
Comparing the two cases in Proposition~\ref{heavytraffic} and
Remark~\ref{comp-remark}, assume that they have the same parameters
($t_1$ and $\mu$) and functions ($\rho_n$ and $\delta_n$), which
satisfies Eqs. (\ref{delta-limit}-\ref{p2}). Then, the
hyper-exponential distributed service time needs a larger number of
servers than the case of exponential distributed service time.
\end{corollary}

\begin{proof}
Using Eq. (\ref{def-sigma}) and Eq. (\ref{def-tau}), we obtain
\begin{equation}
\tau =\frac{\mu^2\sigma^2+c^2}{2\mu
t_1}=\frac{2\sum_{i=1}^k{\left(\frac{P_i}{\mu_i^2}\right)}+(c^2-1)\left(\sum_{i=1}^k{\frac{P_i}{\mu_i}}\right)^2}{2\left(\sum_{i=1}^k{\frac{P_i}{\mu_i}}\right)
t_1}.
\end{equation}

Based on Jensen's Inequality, we can get that
\begin{equation}
\sum_{i=1}^k{\left(\frac{P_i}{\mu_i^2}\right)}\geq
\left(\sum_{i=1}^k{\frac{P_i}{\mu_i}}\right)^2.
\end{equation}

Then,
\begin{equation}
\label{tau} \tau \geq
\frac{(c^2+1)\left(\sum_{i=1}^k{\frac{P_i}{\mu_i}}\right)}{2t_1}=\frac{c^2+1}{2\mu
t_1}.
\end{equation}

Then, in Eq. (\ref{rho-real}), the limit ($\tau$) for
hyper-exponential distributed service time is greater than the limit
($\frac{c^2+1}{2\mu t_1}$) for exponential distributed service time.
Thus, for same $\rho_n$ and $\delta_n$, hyper-exponential
distributed service time needs more servers than exponential
distributed service time.

Consider Eq. (\ref{condition}) which defines the hyper-exponential
service time, we can also get that Eq. (\ref{tau}) achieves equality
if and only if $k=1$.
\end{proof}

Next, we will use the results obtained in Propositions~\ref{hyper}
and \ref{heavytraffic} to compute heavy traffic limits when the
cloud has different QoS requirements. \emph{These results will then
provide guidelines on how many machines to keep active to meet the
QoS requirements of the cloud.}

\section{Applications in Cloud Computing}

\label{cloud} The concept of cloud computing can be traced back to
the 1960s, when John McCarthy claimed that ``computation may someday
be organized as a public utility'' \cite{early}. In recent years,
cloud computing has received increased attention from the industry
\cite{Berkeley}. Many applications of cloud computing, such as
utility computing \cite{uc}, Web 2.0 \cite{Web2_core}, Google app
engine \cite{Google_App}, Amazon web services \cite{EC2,S3} and
Microsoft's Azure services platform \cite{micro_azure}, are widely
used today. Some future application opportunities are also discussed
by Michael Armbrust et al.\ in \cite{Berkeley}. With the rapid
growth of cloud based applications, many definitions, concepts, and
properties of cloud computing have emerged
\cite{Berkeley,Define21,When,CisC,Define}. Cloud computing is an
attractive alternative to the traditional dedicated computing model,
since it makes such services available at a lower cost to the end
users \cite{Berkeley,Market}. In order to provide services at a low
cost, the cost of operating the cloud itself, needs to be kept low.
In \cite{cloudcost}, based on detailed cost analysis of the cloud,
30\% of the ongoing cost is electrical utility costs, and more than
70\% of the ongoing cost is power-related cost which also includes
power distribution and cooling costs. Some typical companies, like
Google, have already claimed that their annual energy costs exceed
their server costs \cite{google_power_cost}. And the power
consumption of Google is 260 million watts
\cite{google_power_amount}. So, power related cost, which is
directly dependent on the {\it number of operational machines} in
the cloud, is a significant fraction of the total cost of operating
a cloud.

\begin{figure}
\centering
\includegraphics[width=0.6\textwidth]{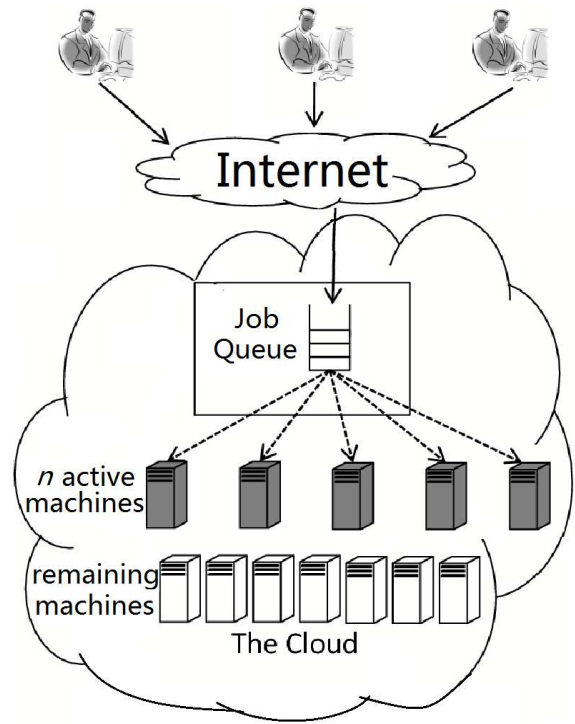}
\caption{Cloud Computing Architecture} \label{arch}
\end{figure}

In \cite{CisC}, P. McFedries points out that clouds are typically
housed in massive buildings and may contain thousands of machines.
This claim is consistent with the fact that large data centers today
often have thousands of machines \cite{Berkeley}. The service system
of a cloud can be viewed as a queueing system. Based on the
stability and efficiency discussions in Section~\ref{intro}, we
focus on the behavior of a cloud in the heavy traffic scenarios.
Figure~\ref{arch} shows the basic architecture. Using the new set of
heavy traffic limit results developed in
Section~\ref{htl_analysis_c2} and \ref{htl_analysis_c3}, we can
achieve the design criteria of power efficient cloud computing
environment, which allows for general and independent arrival
processes and hyper-exponential distributed service times.

\subsection{Heavy Traffic Limits for Different Classes of Clouds}

As discussed earlier, it is important that the cloud operates
stably, which means that the traffic intensity $\rho_n$ should be
less than 1. Further, the cloud also needs to work efficiently,
which means that the traffic intensity $\rho_n$ should be as close
to 1 as possible and should approach 1 as $n\rightarrow \infty$. The
different classes of clouds will result in different heavy traffic
limits, and will thus be governed by different design rules for the
number of operational machines $n$ and traffic intensity $\rho_n$.
From the known literature
\cite{impatient,dim,q1,q2,kingman1965,koll1,koll2}, one can easily
derive the heavy traffic limits for the ZWT and PWT classes. The
derivation is also explicitly shown in our technical report
\cite{Yousi}, and so, here, to save space, we simply state how $n$
and $\rho_n$ should scale to satisfy the QoS requirements of various
clouds.

\subsubsection{ZWT Class} For a cloud of ZWT Class, using
Proposition~\ref{hyper}, we observe that
\begin{equation}
\label{class1} (1-\rho_n)\sqrt{n} \rightarrow \infty,
\end{equation}
from Eqs.\ (\ref{whitt1})-(\ref{whitt3}). If we define $f(n)$ as
$1-\rho_n$, then
\begin{equation}
\begin{split}
\label{new_class1}
&\lim\limits_{n\rightarrow\infty}f(n)=0,\\
&\lim\limits_{n\rightarrow\infty}f(n)\sqrt{n}=\infty.
\end{split}
\end{equation}

\subsubsection{MWT Class}
Applying the result of Proposition~\ref{hyper}, we can show that the
QoS of a cloud of MWT Class can be satisfied if
\begin{equation}
\label{class2}L\le\lim\limits_{n\rightarrow
\infty}(1-\rho_n)\sqrt{n}\le U.
\end{equation}
$U$ and $L$ can be computed from
Eq.~(\ref{upper-bound})--Eq.~(\ref{whitt-bound3}) in
Proposition~\ref{hyper}.

\subsubsection{BWT Class}
We can satisfy the QoS requirement of this class by applying
Proposition~\ref{heavytraffic} to obtain
\begin{equation}
\label{class3}\lim\limits_{n\rightarrow\infty}\frac{(1-\rho_n)n}{-\ln{\delta_n}}=\tau,
\end{equation}
where $\tau$ can be computed by Eq.~(\ref{def-tau}) and
Eq.~(\ref{def-sigma}).

For a cloud of BWT Class, not all functions $\delta_n$, which
decrease to 0, as $n$ goes to infinity, can satisfy the condition.
An appropriate $\delta_n$ that can be used to satisfy the QoS of BWT
Class should satisfy the condition Eq.~(\ref{rho-sqrt-n}) given in
Proposition~\ref{heavytraffic}. Then, the waiting time of jobs for
BWT Class is between 0 and $t$ almost surely as $n\rightarrow
\infty$.

\subsubsection{PWT Class} The QoS requirement of a cloud of PWT Class cloud based on Eq.\
(\ref{kingman}) satisfies
$$P\{W_n\geq t_2\}\sim e^{-\frac{2n\mu (1-\rho)t_2}{\mu^2\sigma^2+c^2}}.$$

For a cloud of PWT Class, to satisfy its QoS requirement, the
traffic intensity must scale as
\begin{equation}
\label{class4}\lim\limits_{n\rightarrow \infty}(1-\rho_n)n=\gamma,
\end{equation}
where
$$\gamma=\frac{-(\mu^2\sigma^2+c^2)\ln{\delta}}{2\mu t_2}.$$

Here, $\mu$ and $\sigma$ are same as Eq.~(\ref{def-sigma}).

\subsection{Number of Operational Machines for Different Classes}
\label{op-num} As discussed in Section~\ref{intro}, an important
motivation of cloud computing is to maximize the workload that the
cloud can support and at the same time satisfy the QoS requirements
of the users. Based on the heavy traffic limits shown in Sections
\ref{htl_analysis_c2} and \ref{htl_analysis_c3}, we have different
heavy traffic limits for different cloud classes (The details of the
ZWT and PWT classes are shown in our technical report \cite{Yousi}).
Thus, in order for the cloud to work efficiently and economically,
we need to compute the least number of machines that the cloud needs
to continue operating for a given QoS requirement.

When $\rho$ is closed to $1$ and $n$ is large, the heavy traffic
limit is a good methodology to approximate the relationship between
$\rho$ and $n$. Based on the heavy traffic limits, we list the
minimum number of machines that the cloud needs to provide under
four classes of clouds, as below.

\begin{itemize}
\item
\underline{The ZWT class:} The $\rho_n$ and $n$ satisfy that
$1-\rho_n \sim f(n)$. Then, the number of operational machines $n$
is $\lceil f^{-1}(1-\rho)\rceil$.

\item
\underline{The MWT class:} The $\rho_n$ and $n$ satisfy that $L\leq
(1-\rho_n)\sqrt{n} \leq U$. Then, for the number of optimal machines
$n$, the lower bound is $\lceil(\frac{L}{1-\rho})^2\rceil$, and the
upper bound is $\lceil(\frac{U}{1-\rho})^2\rceil$.

\item
\underline{The BWT class:} The $\rho_n$ and $n$ satisfy that
$\frac{(1-\rho_n)n}{-\ln{\delta_n}}=\tau$. Then, the number of
operational machines $n$ is $\lceil\frac{\tau
\ln{\delta_n}}{\rho-1}\rceil$.

\item
\underline{The PWT class:} The $\rho_n$ and $n$ satisfy that
$(1-\rho_n)n=\gamma$. Then, the number of operational machines $n$
is $\lceil\frac{\gamma}{1-\rho}\rceil$.

\end{itemize}


Since there are many advanced techniques that can be used to
estimate the parameter $\rho$ and this is not the main focus of this
paper, we assume that the parameter $\rho$ can be estimated from the
data. The number of machines can then be determined by the QoS
requirements and the estimated $\rho$, as shown above.

\section{Numerical Analysis}
\label{evaluation}

\subsection{Evaluation Setup}

We assume that the cloud can accommodate at most $N$ machines.
Clearly, to reduce power consumption, we want to keep the number of
powered servers to a minimum while at the same time satisfying the
corresponding QoS requirements. The parameters for the four classes
are as follows:

\begin{enumerate}
\item
For the ZWT class, we choose $f(n)=n^{-k_1}$, where $k_1=0.25$.
\item For the MWT class, we choose the waiting probability
$\alpha=0.005$.
\item
For the BWT class, we choose $\delta_n=\exp{(-n^{\frac{1}{4}})}$,
which satisfies Eq.\ (\ref{rho-sqrt-n}), and $t_1=0.5$.
\item For the PWT class, we choose the probability threshold $\delta=0.1$ and $t_2=1$.
\end{enumerate}

\subsection{Necessity of Class-based Design}
We first choose a simple process--Poisson process--for arrivals, and
choose an exponential service time distribution (i.e., $\mu=0.3$,
which is the simplest case of the hyper-exponential distribution).

The results characterizing the relationship between the number $n$
of requested machines and the traffic intensity $\rho$ are shown in
Fig.~\ref{totalrho} for $N=10000$. The figure shows that with a
larger pool of machines, not only a large number of jobs, but also a
higher intensity of the offered load can be sustained, especially
for clouds with more stringent QoS requirements.

\begin{figure}
\centering
\includegraphics[width=0.6\textwidth]{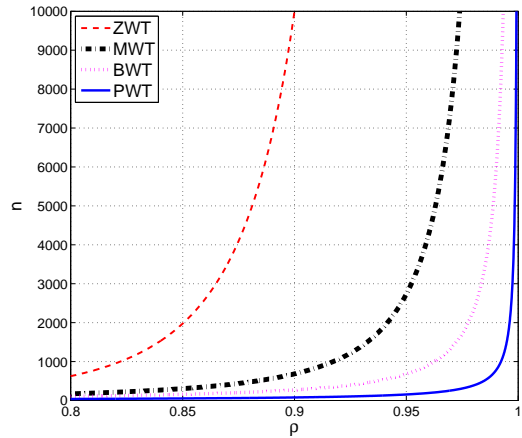}
\caption{Operational Number of Machines for Exponential Distributed
Service Time} \label{totalrho}
\end{figure}

From Fig.~\ref{totalrho}, we can also see that the number of
machines needed for a given value of $\rho$ is quite different for
different QoS classes. Classes with higher QoS require several times
more machines than classes with lower QoS under the same traffic
intensity $\rho$, which implies that different number of operational
machines are necessary for different QoS classes, even for the
simplest case in our scenarios.

In Fig.~\ref{hyperrho}, we now choose a hyper-exponential
distributed service time, with $\bm{\mu}=[1\ 8\ 20]$ and
$\bm{P}=[0.6\ 0.25\ 0.15]$. The results characterizing the
relationship between the number $n$ of requested machines and the
traffic intensity $\rho$ are shown in Fig.~\ref{hyperrho} for
$N=10000$. The figure is similar to the exponential distributed
service time case shown in Fig.~\ref{totalrho}. The difference is
that there are only upper and lower bounds for the MWT class in this
scenario. However, even though there is a certain gap between the
upper and lower bounds for the MWT class, the number of requested
operational machines is still different from other classes when $n$
is large enough.

\begin{figure}
\centering
\includegraphics[width=0.6\textwidth]{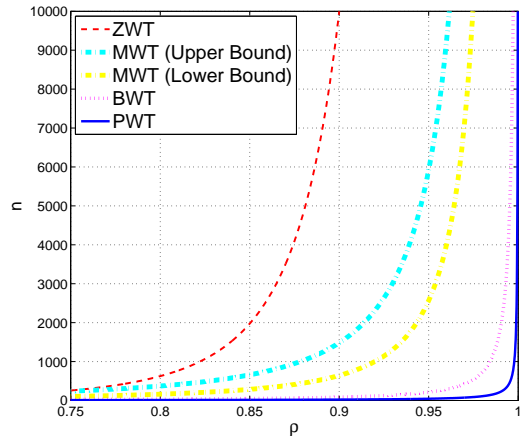}
\caption{Operational Number of Machines for Hyper-exponential
Distributed Service Time}\label{hyperrho}
\end{figure}

Note that Figs.~\ref{totalrho} and \ref{hyperrho} can also be used
to find the maximal traffic intensity a cloud can support while
satisfying a given QoS requirement for a given number of machines in
the cloud.

\begin{figure}
\centering
\includegraphics[width=0.6\textwidth]{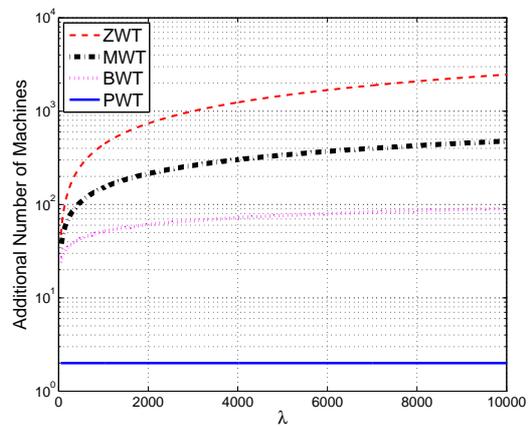}
\caption{Additional Operational Machines for Exponential Distributed
Service Time} \label{totalrho_lambda}
\end{figure}

\begin{figure}
\centering
\includegraphics[width=0.6\textwidth]{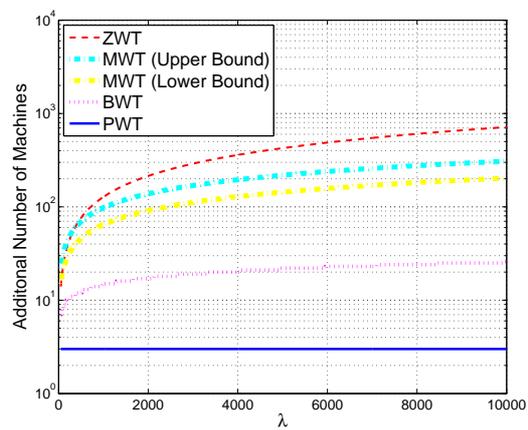}
\caption{Additional Operational Machines for Hyper-exponential
Distributed Service Time}\label{hyperrho_lambda}
\end{figure}

Given an arrival rate $\lambda$, the basic request number of
machines is equal to $\frac{\lambda}{\mu}$. However, it is not
enough to satisfy the different QoS requirements. For different QoS
requirements, the corresponding number of machines are shown in
Figs.~\ref{totalrho_lambda} and \ref{hyperrho_lambda}.
Figs.~\ref{totalrho_lambda} and \ref{hyperrho_lambda} are under the
same scenarios as Figs.~\ref{totalrho} and \ref{hyperrho}
correspondingly. From these two figures, we can see that, for the
same arrival rate, different classes need different additional
number of machines to satisfy different QoS requirements. Similarly,
given a arrival rate $\lambda$, the heaviest traffic intensity the
system can support under a given QoS requirement is shown in
Figs.~\ref{totalrho_rho} and \ref{hyperrho_rho}.

\begin{figure}
\centering
\includegraphics[width=0.6\textwidth]{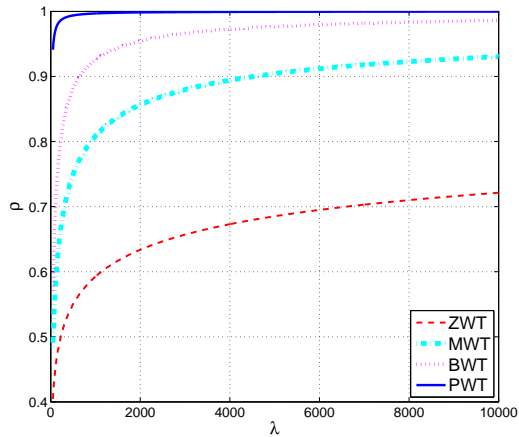}
\caption{Traffic Intensity for Exponential Distributed Service Time}
\label{totalrho_rho}
\end{figure}

\begin{figure}
\centering
\includegraphics[width=0.6\textwidth]{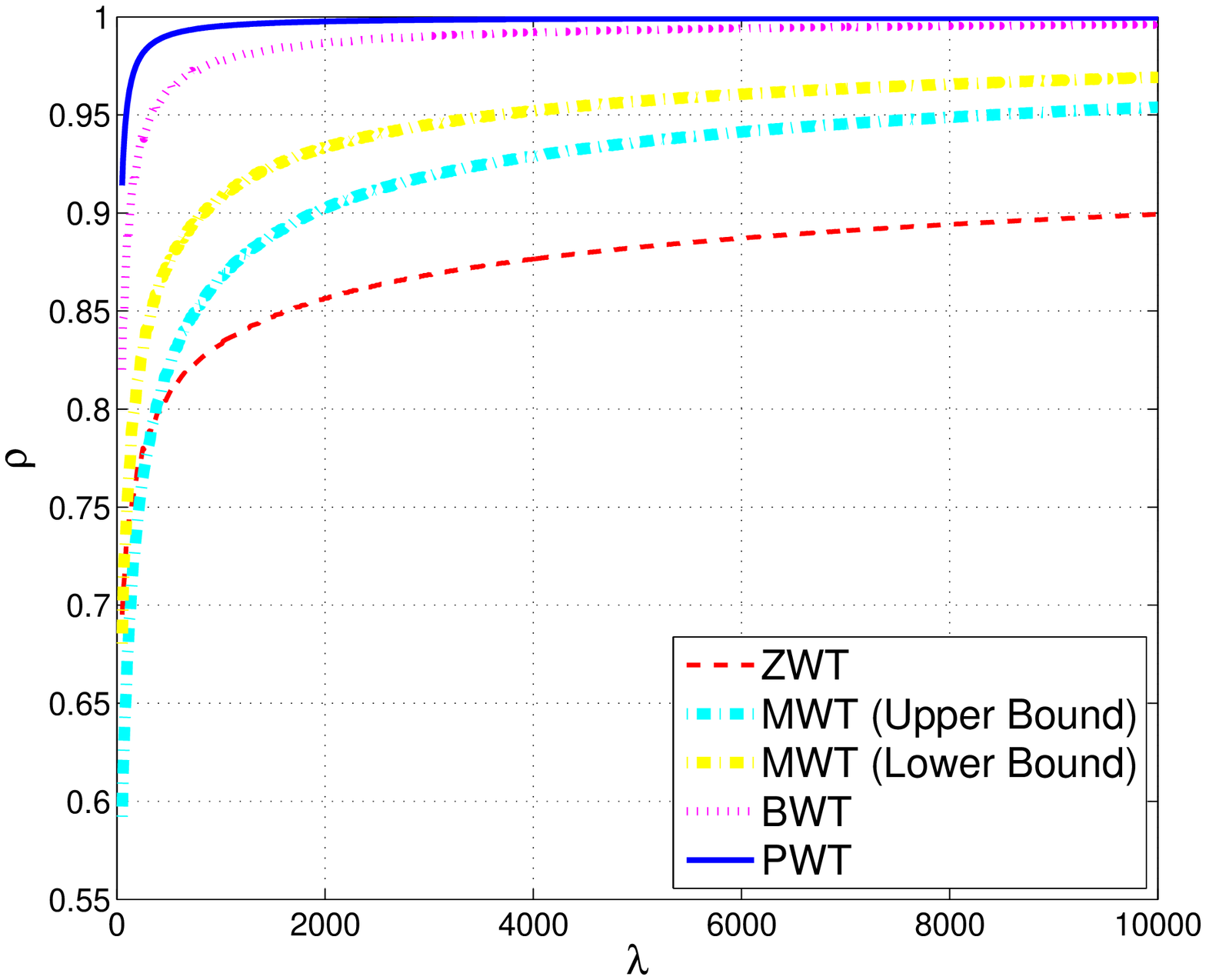}
\caption{Traffic Intensity for Hyper-exponential Distributed Service
Time}\label{hyperrho_rho}
\end{figure}

\subsection{Evaluation for the MWT and BWT Classes}

\begin{figure}
\centering
\includegraphics[width=0.6\textwidth]{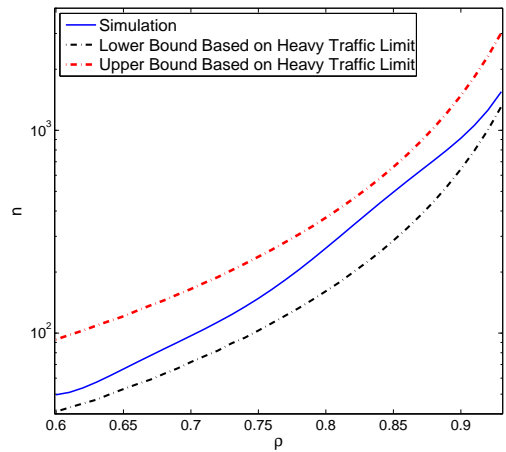}
\caption{Simulation results for the queueing systems of the MWT
Class (Log Y-Axis)} \label{hyper_log}
\end{figure}

For the MWT class, we also choose the same distribution of service
time as above (i.e., $\bm{\mu}=[1\ 8\ 20]$ and $\bm{P}=[0.6\ 0.25\
0.15]$) as an example. The performance of the MWT class is shown in
Fig.~\ref{hyper_log}.

We define a ratio to evaluate the tightness of the upper and lower
bounds for the clouds of the MWT class as below.
\begin{equation}
r\triangleq \frac{U}{L}=r_1r_2,
\end{equation}
where
\begin{equation}
\begin{split}
r_1&=\frac{\psi_U}{\psi_L},\\
r_2&=\frac{\sum\limits_{i=1}^k{\left((1+\frac{c^2-1}{2}P_i)\sqrt{\frac{P_i}{\mu_i}}\right)}}{\max\limits_{i\in{\{1,...k\}}}
\left\{(1+\frac{c^2-1}{2}P_i)\sqrt{{\frac{P_i}{\mu_i}}}\right\}}.
\end{split}
\end{equation}

For a given $k$, $r_1$ and $r_2$ are independent. $r_2$ is
determined by how the sum
$\sum\limits_{i=1}^k{\left((1+\frac{c^2-1}{2}P_i)\sqrt{\frac{P_i}{\mu_i}}\right)}$
dominates the largest item $\max\limits_{i\in{\{1,...k\}}}
\left\{(1+\frac{c^2-1}{2}P_i)\sqrt{{\frac{P_i}{\mu_i}}}\right\}$.
Its domain is interval $[1,k]$. $r_1$ is determined by parameter $k$
and $\alpha$, and is independent of $\bm{P}$ and $\bm{\mu}$. For
different values of $k$ and $\alpha$, the corresponding ratio $r_1$
is shown in Fig.~\ref{r1}. From Fig.~\ref{r1}, we can see that $r_1$
is typically a small constant, even when $\alpha$ and $k$ are large
(e.g. if $\alpha=0.15$ and $k=20$, then $r_1$ is less than 2).

\begin{figure}
\centering
\includegraphics[width=0.6\textwidth]{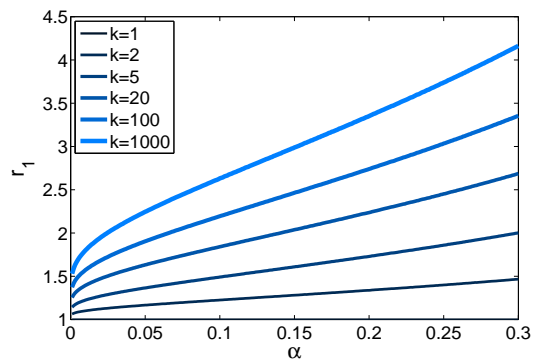}
 \caption{Ratio $r_1$} \label{r1}
\end{figure}

The performance of the BWT class is shown in
Fig.~\ref{hyper_class3}.

\begin{figure}
\centering
\includegraphics[width=0.6\textwidth]{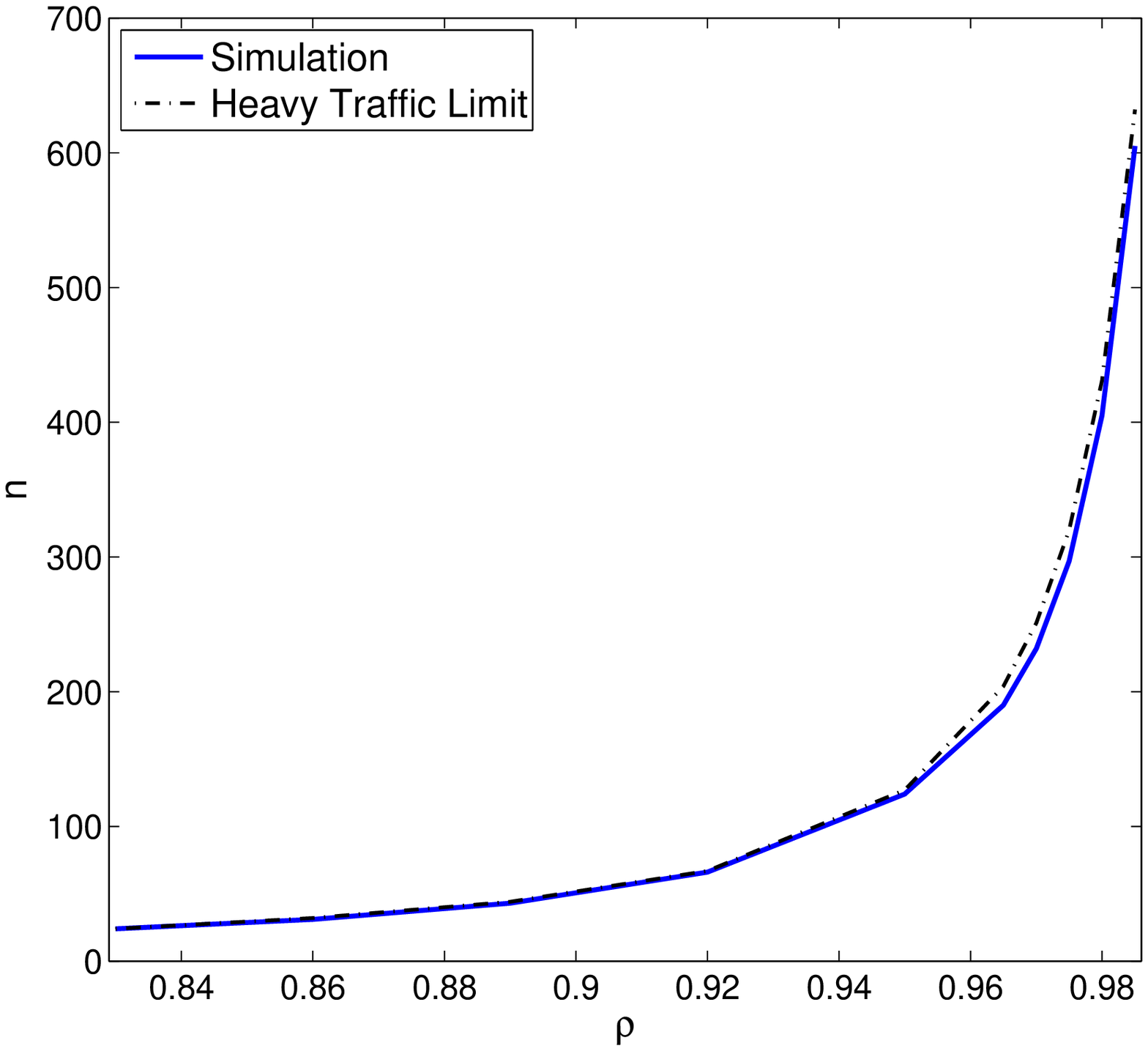}
\caption{Simulation results for the queueing systems of the BWT
Class} \label{hyper_class3}
\end{figure}


For non-Poisson arrival processes, we also select 2-state Erlang
distribution and deterministic distribution as examples. The
simulation results for the MWT and BWT classes are shown in
Figs.~\ref{fig:other_MWT} and \ref{fig:other_BWT}.

\begin{figure*}
\centering

\centerline{\subfigure[The MWT Class with Erlang Arrivals ($Er_2$)]{
\includegraphics[width=0.5\textwidth]{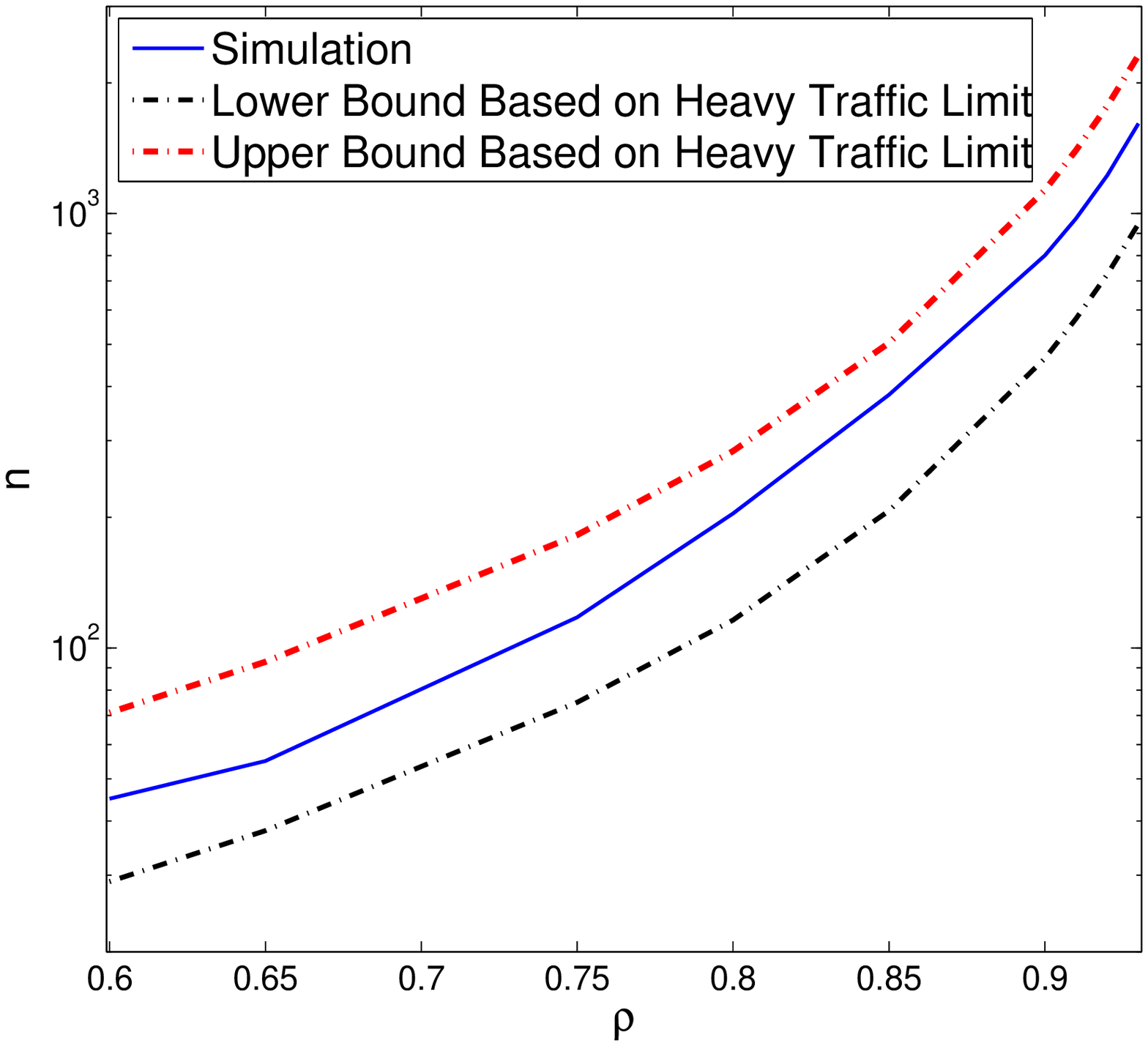}
\label{fig:class2_erlang} }\hspace{0.01\textwidth} \subfigure[The
MWT class with Deterministic Arrivals]{
\includegraphics[width=0.5\textwidth]{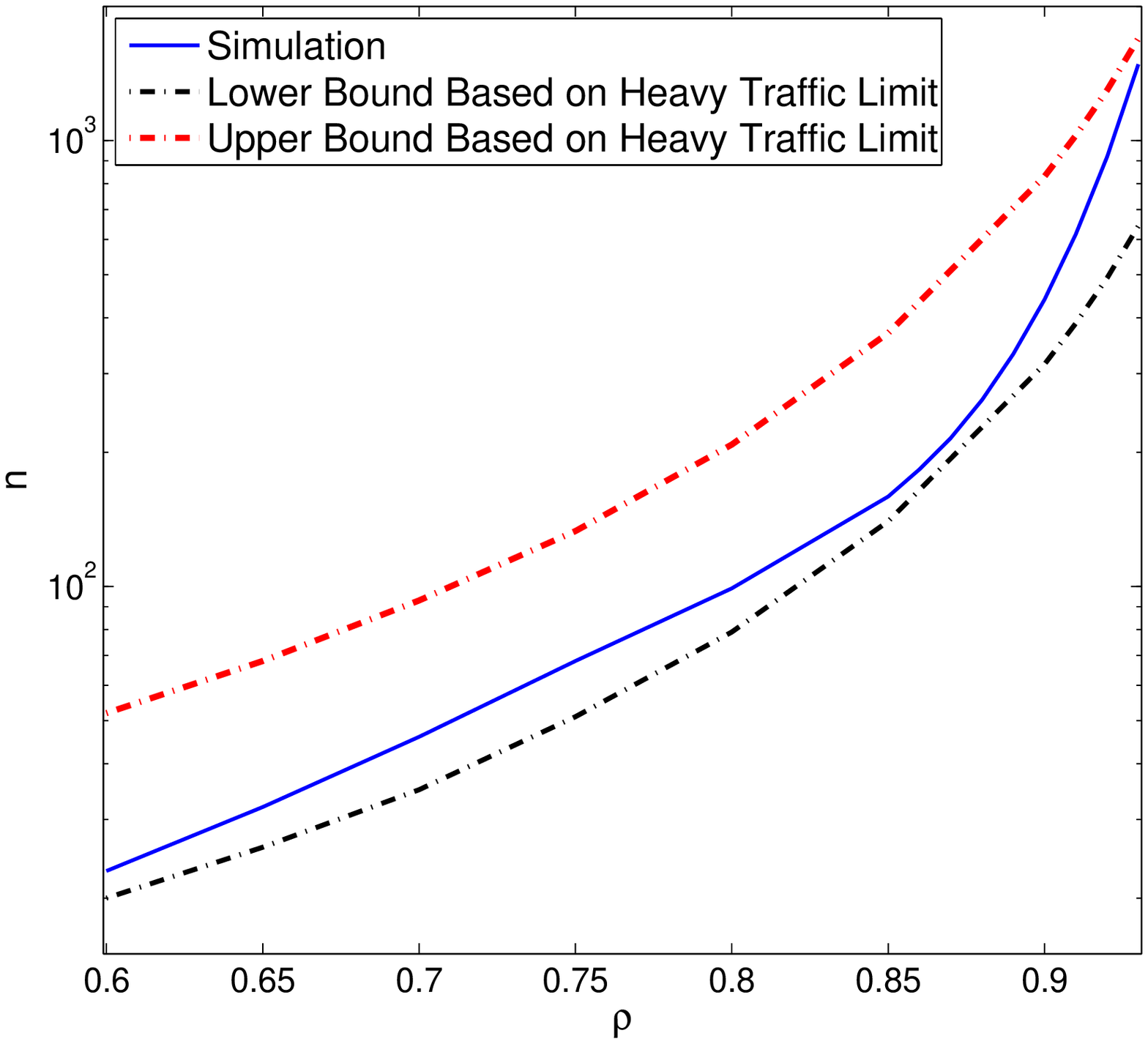}
\label{fig:class2_d} } } \caption{Simulation Results for the MWT
class with Other Arrival Processes (Log Y-Axis)}
\label{fig:other_MWT}
\end{figure*}

\begin{figure*}
\centering \centerline{\subfigure[The BWT class with Erlang Arrivals
($Er_2$)]{
\includegraphics[width=0.5\textwidth]{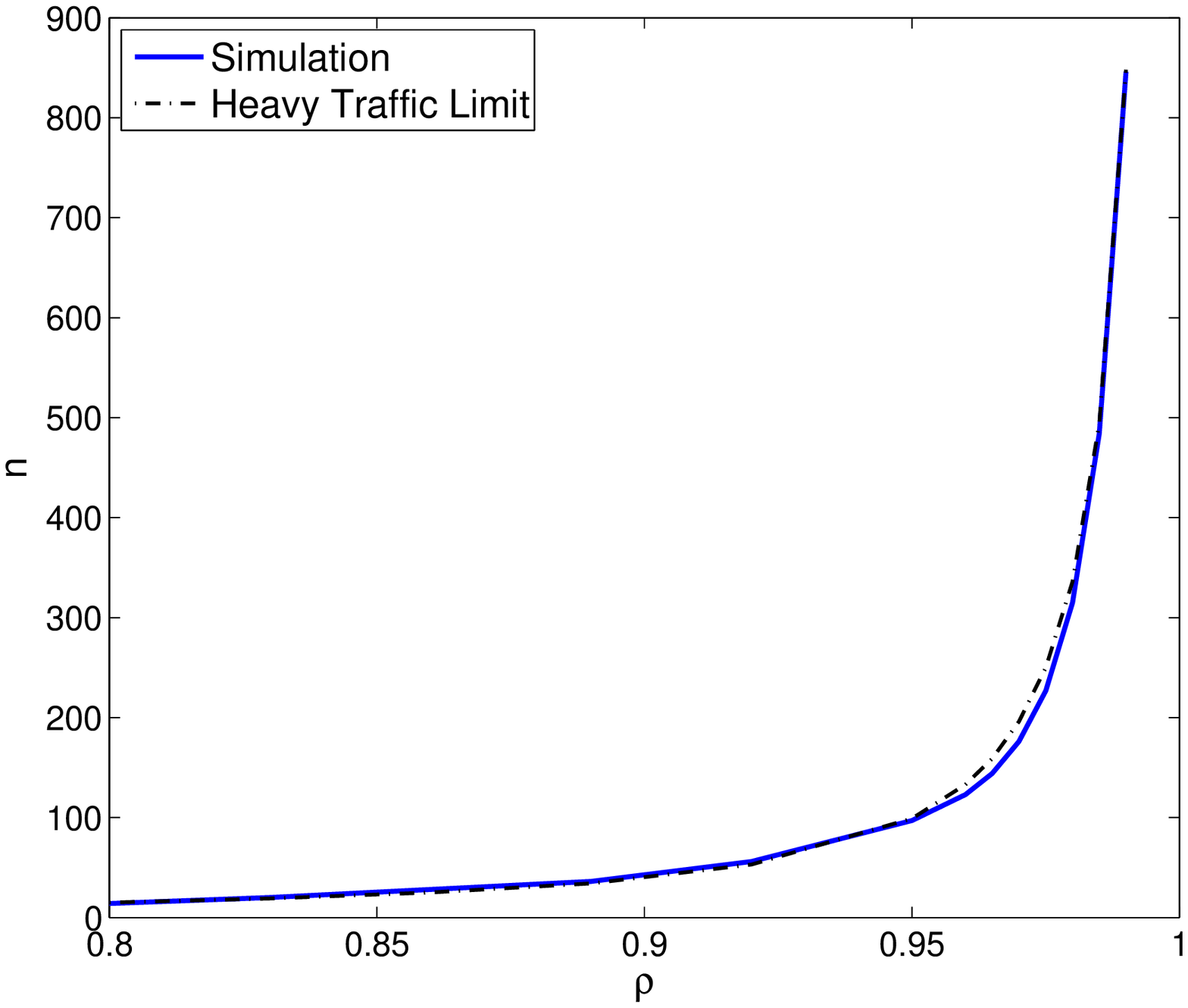}
\label{fig:class3_erlang} } \hspace{0.01\textwidth} \subfigure[The
BWT Class with Deterministic Arrivals]{
\includegraphics[width=0.5\textwidth]{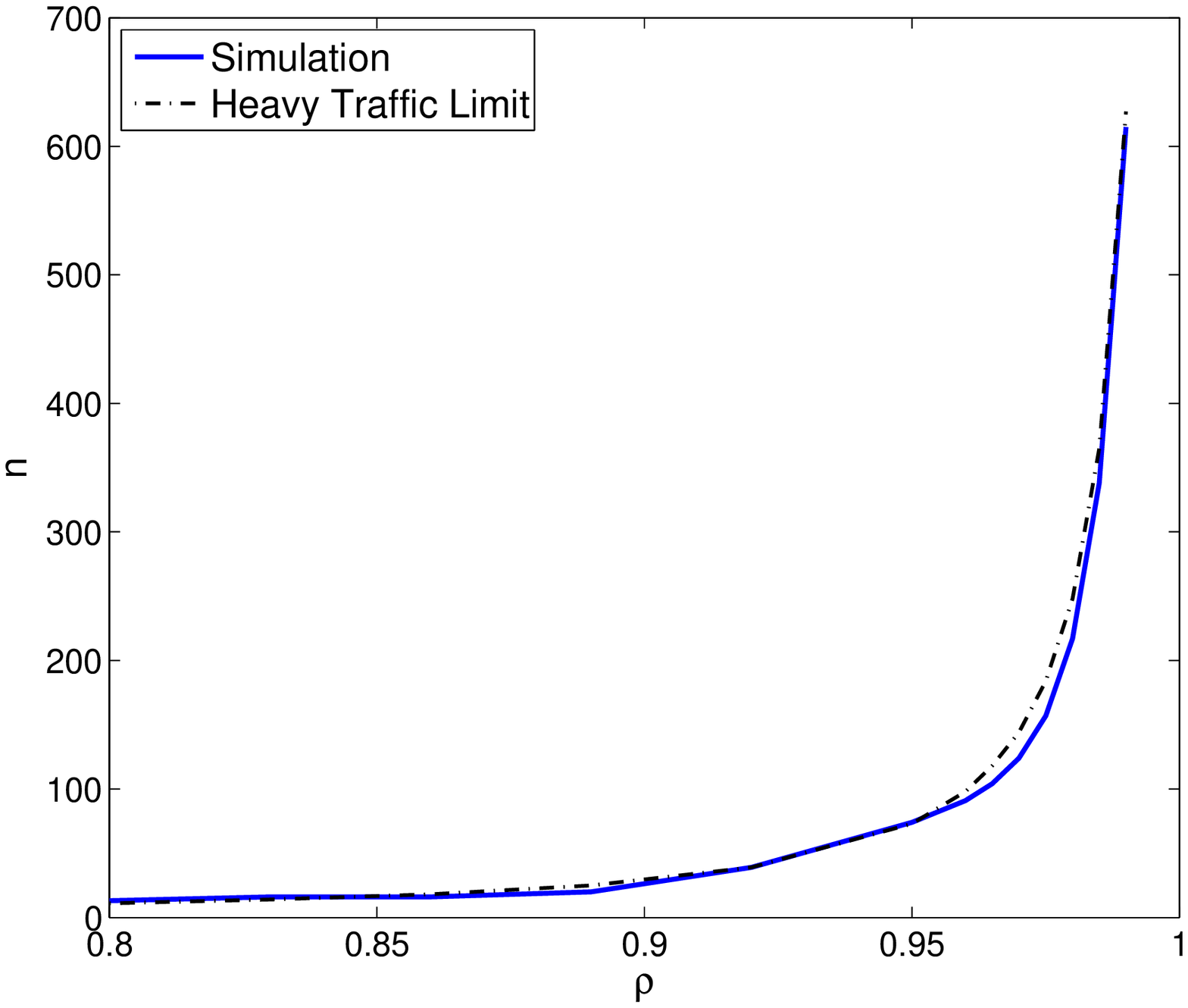} \label{fig:class3_d} } }

\caption{Simulation Results for the BWT class with Other Arrival
Processes } \vspace{-1em} \label{fig:other_BWT}
\end{figure*}

We have used the heavy traffic limit results to design the cloud for
finite values of $n$ in Figs.~\ref{hyper_log}, \ref{hyper_class3},
\ref{fig:other_MWT} and \ref{fig:other_BWT}. From these figures, we
observe that the simulation results closely follow the result
obtained from the heavy traffic analysis even when the number of
machines is not very large (e.g., only 100) and traffic is not very
heavy (e.g., $\rho=0.85$).

\section{Conclusion}
\label{conclude} In this paper, we study the heavy traffic limits of
GI/H/n queues. First, we classify the queueing systems into four
classes based on the QoS requirements. Then, we develop heavy
traffic limits that characterize the performance of the queueing
systems for different types of QoS requirements. For the MWT and BWT
classes, new heavy traffic limits are derived. Based the analysis of
heavy traffic limits for different classes in this paper and
existing results, we show the relationship between heavy traffic
limits and QoS requirements to obtain design rules in the cloud
computing environment as an application. The numerical results show
that different rules are necessary for computing the number of
operational machines for different cloud classes, and show the
performance of the heavy traffic limits when the number of
operational machines is finite. In the future, we plan to extend our
work to jobs that need multiple servers or multiple stages, and
apply it to improve widely-used frameworks, such as MapReduce.

\bibliographystyle{spmpsci}

\bibliography{Group}

\begin{thebibliography}{10}

\bibitem{call_center}
I.~Mitchell, ``Call center consolidation ¡ª does it still make sense?,'' {\em
  Business Communications Review}, pp.~24--28, December 2001.

\bibitem{google_platform}
\url{http://en.wikipedia.org/wiki/Google_platform}.

\bibitem{iglehart1965}
D.~L. Iglehart, ``Limiting diffusion approximations for the many server queue
  and the repairman problem,'' {\em Journal of Applied Probability}, vol.~2,
  pp.~429--441, December 1965.

\bibitem{whitt1981}
S.~Halfin and W.~Whitt, ``Heavy-traffic limits for queues with many exponential
  servers,'' {\em Operations Research}, vol.~29, pp.~567--588, May-June 1981.

\bibitem{whitt_book}
W.~Whitt, {\em Stochastic-Process Limits}.
\newblock New York, NY, USA: Springer-Verlag New York, Inc., 2002.

\bibitem{ph}
A.~A. Puhalskii and M.~I. Reiman, ``{The Multiclass GI/PH/N Queue in the
  Halfin-Whitt Regime},'' {\em Advances in Applied Probability}, vol.~32,
  pp.~564--595, June 2000.

\bibitem{reed}
J.~Reed, ``{The G/GI/n Queue in the Halfin-Whitt Regime},'' {\em The Annals of
  Applied Probability}, vol.~19, no.~6, pp.~2211--2269, 2009.

\bibitem{steady}
D.~Gamarnik and P.~Momcilovic, ``{Steady-state analysis of a multi-server queue
  in the Halfin-Whitt regime},'' {\em Advances in Applied Probability},
  vol.~40, no.~2, pp.~548--577, 2008.

\bibitem{whitt2005}
W.~Whitt, ``{Heavy-Traffic Limits for the $G/H_2^*/n/m$ Queue},'' {\em
  Mathematics of Operations Research}, vol.~30, pp.~1--27, February 2005.

\bibitem{diffusion}
W.~Whitt, ``{A Diffusion Approximation for the G/GI/n/m Queue},'' {\em
  Operations Research}, vol.~52, pp.~922--941, November-December 2004.

\bibitem{impatient}
O.~Garnett, A.~Mandelbaum, and M.~I. Reiman, ``Designing a call center with
  impatient customers,'' {\em Manufacturing \& Service Operations Management},
  vol.~4, pp.~208--227, Summer 2002.

\bibitem{dim}
S.~Borst, A.~Mandelbaum, and M.~I. Reiman, ``Dimensioning large call centers,''
  {\em Operations Research}, vol.~52, pp.~17--34, Janurary-February 2004.

\bibitem{dnc}
A.~Bassamboo, J.~M. Harrison, and A.~Zeevi, ``{Design and Control of a Large
  Call Center: Asymptotic Analysis of an LP-based Method},'' {\em Operations
  Research}, vol.~54, pp.~419--435, May-June 2006.

\bibitem{level}
M.~Armony, I.~Gurvich, and A.~Mandelbaum, ``{Service Level Differentiation in
  Call Centers with Fully Flexible Servers},'' {\em Management Science},
  vol.~54, pp.~279--294, February 2008.

\bibitem{QED}
N.~Gans, G.~Koole, and A.~Mandelbaum, ``{Telephone Call Centers: Tutorial,
  Review, and Research Prospects},'' {\em {Manufacturing \& Service Operations
  Management}}, vol.~5, no.~2, pp.~79--141, 2003.

\bibitem{Berkeley}
M.~Armbrust, A.~Fox, R.~Griffith, A.~Joseph, R.~Katz, A.~Konwinski, G.~Lee,
  D.~Patterson, A.~Rabkin, I.~Stoica, and M.~Zaharia, ``Above the clouds: A
  berkeley view of cloud computing,'' tech. rep., UC Berkeley, February 2009.

\bibitem{Market}
R.~Buyya, C.~S. Yeo, and S.~Venugopal, ``Market-oriented cloud computing:
  Vision, hype, and reality for delivering it services as computing
  utilities,'' in {\em Proceedings of High Performance Computing and
  Communications, 2008. HPCC'08}, pp.~5--13, September 2008.

\bibitem{Define}
L.~M. Vaquero, L.~Rodero-Merino, J.~Caceres, and M.~Lindner, ``A break in the
  clouds: Towards a cloud definition,'' {\em AMC SIGCOMM Computer Communication
  Review}, vol.~39, pp.~50--55, January 2009.

\bibitem{Yousi}
Y.~Zheng, N.~Shroff, P.~Sinha, and J.~Tan, ``{Design of a power efficient cloud
  computing environment: heavy traffic limits and QoS},'' tech. rep., Ohio
  State University, February 2011.

\bibitem{schedule::tail}
M.~Lin, A.~Wierman, L.~Andrew, and E.~Thereska, ``Jayakrishnan nair and adam
  wierman and bert zwart,'' in {\em 48th Annual Allerton Conference on
  Communication, Control, and Computing, 2010. Allerton'10}, pp.~969--976,
  Feburary 2010.

\bibitem{schedule::possible}
A.~Wierman and B.~Zwart, ``Is tail-optimal scheduling possible?.''

\bibitem{schedule::control}
A.~L. Stolyar, ``Control of end-to-end delay tails in a multiclass network:
  Lwdf discipline optimality,'' {\em Business Communications Review},
  pp.~1151--1206, 2003.

\bibitem{schedule::large}
A.~L. Stolyar and K.~Ramanan, ``Largest weighted delay first scheduling: large
  deviations and optimality,'' {\em The Annals of Applied Probability},
  vol.~11, no.~1, pp.~1--48, 2001.

\bibitem{kingman_book}
J.~F.~C. Kingman, {\em Poisson Processes}.
\newblock New York, NY, USA: Oxford University Press, USA, 1993.

\bibitem{daley}
D.~J. Daley, ``Certain optimality properties of the first come first served
  discipline for g/g/s queues,'' {\em Stochastic Processes and their
  Applications}, vol.~25, pp.~301--308, 1987.

\bibitem{whitt84}
W.~Whitt, ``The amount of overtakingin a network of queues,'' {\em Stochastic
  Processes and their Applications}, vol.~14, pp.~411--426, 1984.

\bibitem{towsley}
Z.~Liu and D.~Towsley, ``Stochastic scheduling in in-forest networks,'' {\em
  Advances of Applied Probability}, vol.~26, pp.~222--241, 1994.

\bibitem{fcfs}
S.~G. Foss and N.~I. Chernova, ``On optimality of the fcfs discipline in
  multiserver queueing systems and networks,'' {\em Siberian Mathematical
  Journal}, vol.~42, pp.~372--385, March-April 2001.

\bibitem{sp}
S.~Ross, {\em Stochastic Processes}.
\newblock New York, NY, USA: John Wiley \& Sons, 1996.

\bibitem{pnm}
P.~BIllingsley, {\em Probability and Measure}.
\newblock New York, NY, USA: John Wiley \& Sons, 1995.

\bibitem{kingman1965}
J.~F.~C. Kingman, ``The heavy traffic approximation in the theory of queues,''
  in {\em Proceedings of Symposium on Congestion Theory}, pp.~137--159, 1965.

\bibitem{koll1}
J.~Kollerstrom, ``{Heavy Traffic Theory for Queues with Several Servers. I},''
  {\em Journal of Applied Probability}, vol.~11, pp.~544--552, September 1974.

\bibitem{koll2}
J.~Kollerstrom, ``{Heavy Traffic Theory for Queues with Several Servers. II},''
  {\em Journal of Applied Probability}, vol.~16, pp.~393--401, June 1979.

\bibitem{early}
J.~McCarthy, ``Mit centennial,'' 1961.

\bibitem{uc}
J.~Broberg, S.~Venugopal, and R.~Buyya, ``Market-oriented grids and utility
  computing: The state-of-the-art and future directions,'' {\em Journal of Grid
  Computing}, vol.~6, pp.~255--276, September 2008.

\bibitem{Web2_core}
P.~Sharma, ``Core characteristics of web 2.0 services.''
  \url{http://www.techpluto.com/web-20-services/}, November 2008.

\bibitem{Google_App}
\url{http://code.google.com/appengine}.

\bibitem{EC2}
\url{http://aws.amazon.com/ec2}.

\bibitem{S3}
\url{http://aws.amazon.com/s3}.

\bibitem{micro_azure}
\url{http://www.microsoft.com/windowsazure/}.

\bibitem{Define21}
J.~Geelan, ``Twenty one experts define cloud computing,'' {\em Virtualization,
  Electronic Magazine}, August 2008.

\bibitem{When}
R.~Bragg, ``Cloud computing: When computers really rule,'' {\em Tech News
  World, Electronic Magazine}, July 2008.

\bibitem{CisC}
P.~McFedries, ``The cloud is the computer,'' {\em IEEE Spectrum Online,
  Electronic Magazine}, August 2008.

\bibitem{cloudcost}
A.~Greenberg, J.~Hamilton, D.~A. Maltz, and P.~Patel, ``The cost of a cloud:
  research problems in data center networks,'' {\em ACM SIGCOMM Computer
  Communication Review}, vol.~39, pp.~68--73, Janurary 2009.

\bibitem{google_power_cost}
\url{http://greenit.net/whygreenit.html}.

\bibitem{google_power_amount}
J.~Glanz, ``Google details, and defends, its use of electricity.''
  \url{http://www.nytimes.com/2011/09/09/technology/google-details-and-defends%
-its-use-of-electricity.html}, September 2011.

\bibitem{q1}
L.~Kleinrock, {\em {Queueing Systems, Volumn 1: Theory}}.
\newblock New York, NY, USA: John Wiley \& Sons, 1975.

\bibitem{q2}
L.~Kleinrock, {\em {Queueing Systems, Volumn 2: Computer Applications}}.
\newblock New York, NY, USA: John Wiley \& Sons, 1976.

\end{thebibliography}

\end{document}